\documentclass[journal,draftcls,onecolumn,12pt,twoside]{IEEEtranTCOM}

\normalsize

\usepackage[cmex10]{amsmath}
\usepackage{amssymb}
\usepackage{amsthm}
\usepackage{tabularx}
\usepackage{graphicx}
\usepackage{graphics}
\usepackage{ifthen}

\usepackage[font=footnotesize]{subfig}
\usepackage{flushend}
\usepackage{multicol}
\usepackage{float}

\newcommand{\bs}{\boldsymbol}
\newcommand{\refeq}[1]{(\ref{#1})}
\newcommand{\reffig}[1]{Fig.~\ref{#1}}
\newcommand{\refsec}[1]{Sec.~\ref{#1}}
\newcommand{\refthm}[1]{Theorem~\ref{#1}}

\newcommand{\reflm}[1]{Lemma~\ref{#1}}
\newcommand{\refapp}[1]{Appendix~\ref{#1}}
\newcommand{\refpro}[1]{Proposition~\ref{#1}}

\newtheorem{definition}{Definition}
\newtheorem{lemma}{Lemma}
\newtheorem{proposition}{Proposition}
\newtheorem{theorem}{Theorem}

\newcommand{\Z}{\mathbb{Z}}
\newcommand{\R}{\mathbb{R}}
\newcommand{\Q}{\mathbb{Q}}

\newcommand{\N}{\mathcal{N}}

\begin{document}
%
\title{Diversity of Low-Density Lattices}
\author{Mayur~Punekar,~\IEEEmembership{Member,~IEEE,}
        and~Joseph~Jean~Boutros,~\IEEEmembership{Senior~Member,~IEEE}%
\thanks{M. Punekar and J. J. Boutros are with the Department
of Electrical and Computer Engineering, Texas A\&M University at Qatar, Education City,
23874 Doha, Qatar. e-mail: \{mayur.punekar, joseph.boutros\}@ieee.org .}
\thanks{Manuscript received ; revised . 
The material in this paper was presented in part at the International Conference on Telecommunications (ICT), Sydney, Australia, April 2015 and in part at the IEEE International Symposium on Information Theory (ISIT), Hong Kong, June 2015.} }

\markboth{To be Submitted to IEEE Transactions on Communications}%
{Submitted paper}

\maketitle

\vspace*{-20pt}
\begin{abstract}
The non-ergodic fading channel is a useful model for various wireless communication channels in both indoor and outdoor environments. 
Building on Poltyrev's work on infinite lattice constellations for the Gaussian channel, 
we derive a Poltyrev outage limit (POL) for lattices in presence of block fading.
We prove that the diversity order of this POL is equal to the number of degrees of freedom in the channel.
Further, we describe full-diversity constructions of real lattices defined by their integer-check matrix, i.e., the inverse of their generator matrix.
In the first construction suited to maximum-likelihood decoding, 
these lattices are defined by sparse (low-density) or non-sparse 
integer-check matrices. 
Based on a special structure of the lattice binary image, a second full-diversity lattice construction is described for sparse integer-check matrices in the  context  of  iterative   
probabilistic decoding. Full diversity is theoretically proved in both cases. We also propose a method to construct lattices for diversity order $4$ suitable for iterative probabilistic decoding.
Computer simulation results for diversity order $L=2$ and $L=4$ confirm that the proposed low-density lattices attain the maximal diversity order. 
The newly defined POL is used during this simulation to declare an outage error without decoding, which drastically improves the decoding runtime.

\end{abstract}
%
\IEEEpeerreviewmaketitle

\section{Introduction}
Coded modulations are nowadays an integral part of almost all communication systems.
Channel coding and linear modulation signal sets can be combined together 
in many different ways \cite{Proakis2008}. Lattices in the real Euclidean space 
are a special case of coded modulations. Lattices are infinite constellation of points 
with a group structure \cite{Conway1998}. 
With lattice points perturbed by additive white Gaussian noise, Poltyrev showed that there exists a lattice that can be correctly decoded if and only if the noise variance is less than a threshold \cite{Poltyrev1994}.
This threshold will be called {\em Poltyrev threshold} in the sequel and
corresponds to data transmission at infinite rate. 
Its counterpart for finite-rate data transmission is Shannon capacity \cite{Cover2006}.
Poltyrev's work does not restrict the dimension of the infinite constellation that achieves a vanishing error probability. 
Infinite constellations with a finite dimension over the Gaussian channel were analyzed by Ingber {\it et al.} in \cite{Ingber2013}. 

There are various methods proposed in the literature to build lattices, e.g., lattices from error-correcting codes \cite{Leech1971}, lattices built from groups \cite{Conway1998}, algebraic constructions of lattices \cite{Boutros1996}, etc. \textit{Low-density  lattice codes} (LDLC) are  a special class  of lattices  proposed by Sommer  {\it et  al.}   in~\cite{Sommer2008}  which  can  be  decoded  by  iterative probabilistic   message  passing.    Another  family   of  iteratively decodable lattices has been  published in~\cite{Boutros2014}, where authors described a generalized low-density construction. In the present work, we focus on LDLC without using any shaping region.
%
%

%
%
We  consider transmission  using lattices  over a  general non-ergodic fading   (block-fading)   channels  for   single-input   single-output systems. For such a channel model a codeword of length $n$ is divided into $L$  blocks of equal length  $n/L$ such that the  fading within a given  block is  the same,  whereas  it is  different and  independent across different blocks. Such a channel has $L$ degrees of freedom. 
Poltyrev threshold is not valid for such a channel as the outage probability depends not only signal-to-noise ratio (SNR) but also on the channel coefficients. 
If a  lattice code  is used  for  transmission over  such a  block-fading
channel and  the error  probability at  the output  of the  decoder is
proportional to  $1/\gamma^L$ (where  $\gamma$ is  the SNR) then such a lattice code is said to have diversity order $L$ or
referred to  as full-diversity lattice. However,  randomly constructed
lattice codes do not have full-diversity property. To the  best of
the authors' knowledge  no effort has been made in the  literature to design
full-diversity   lattice    codes.   The diversity-multiplexing tradeoff
(DMT) of infinite constellations for multiple-antenna fading channels was 
studied in ~\cite{Yona2014}. The DMT is valid only at high signal-to-noise ratio
and at high rate per real dimension. 

In this paper, first we derive a {\em Poltyrev outage limit} (POL) for lattices over block-fading channels using Poltyrev threshold. We prove that POL has diversity $L$ for a channel with $L$ independent block fadings, i.e., POL has full diversity. One of the most important applications of POL is to detect inadmissible channel states. Inadmissible fadings are deep fadings that cause an outage event. Hence, POL can be used to declare error without decoding the received point when decoder would not be able to decode the transmitted point due to deep fading.
Further, we propose methods  to construct  lattices which exhibit full-diversity when  decoded by maximum-likelihood (ML)  and iterative decoder. 
We give conditions for a lattice to be full-diversity under ML decoding  and  propose a full-diversity  construction  method  for lattice codes as stated by \refthm{thm_fulldiv_ml}. 
This construction method can be used to generate sparse and non-sparse  lattice codes  suitable for ML  decoding. 
Subsequently, the random  LDLC ensemble is  analyzed through the low-density parity-check code derived from its binary image. Following this analysis, a  theorem  for building full-diversity LDLC  is stated. 
This theorem is valid for iterative  probabilistic decoding of LDLC with a sparse
integer-check matrix.
Apart from above mentioned methods, one more method to construct double-diversity LDLC is presented which is valid for both, the ML decoding and also for iterative probabilistic decoding. A generalization of this method is proposed to build LDLC with diversity order $L=4$ for iterative probabilistic decoding.

Simulation results for double-diversity LDLC designed for ML decoding and LDLC constructed for the diversity order $L=2$ and $L=4$ for iterative decoding are also provided which validate our theoretical results.
In our simulations we utilize POL to detect inadmissible channel states and declare error without decoding the received point which reduces the decoder runtime. We demonstrate that the ML decoding of full-diversity LDLC in dimension $n=64$ is not feasible without using POL to declare errors without decoding.

%

The  rest of  the paper  is structured  as follows. Background information, notation and the communication model considered in this paper is presented in Section~II. 
Poltyrev outage limit for lattices over block fading channels is derived in Section III. 
Section~IV gives the proof for the diversity order of the Poltyrev outage limit. 
In Section~V, we propose the full-diversity lattice code suitable for ML decoding.  Full-diversity LDLC  suitable  for  iterative decoding  are proposed in Section~VI.  Section~VII  discusses  simulation results.  We conclude the paper in Section~VIII.

\section{Background and Notation}
\subsection{Lattices} \label{sec:lattice_code}
Let $\R$ be the field of real numbers, $\Z$ its ring of integers, and $\Q$ the field of fractions of $\Z$.
A {\em lattice} $\Lambda \subset \R^n$, also called a point lattice, is a free $\Z$-module of rank $n$ in $\R^n$. An element belonging to $\Lambda$ is called a {\em point} or equivalently a vector. Any point 
$\bs{x}~=~(x_1, x_2,\ldots, x_n)^T \in \Lambda$ can be written as an integer linear combination of $n$ points 
 ${\bs{x}}=\sum_{i=1}^n z_1 {\bs{v}}_i,$ 
where $\{ {\bs{v}}_i\}$ is a $\Z$-basis of $\Lambda$, $v_{ij} \in \R$, and $z_1 \in \Z$. The $n \times n$ matrix $\mathbf{G}$ built from a basis is a {\em generator matrix} for $\Lambda$. In column convention, let $\mathbf{G}=[v_{ij}]$, then a lattice point is written as $\bs{x}=\mathbf{G}\bs{z}$, where $\bs{z} \in \Z^n$. The {\em fundamental volume} of the lattice $\Lambda$ is given by $|\det(\mathbf{G})|$. 
For more information on lattices, we refer the reader to \cite{Conway1998}.

Poltyrev threshold can be stated as follows: 
there exists a lattice $\Lambda$ of high enough dimension $n$ for which the transmission error probability 
over an additive white Gaussian noise (AWGN) channel can be reduced to an arbitrary low level 
if and only if $\sigma^2 < \sigma_{max}^2$ \cite{Poltyrev1994}\cite{Sommer2008} 
where $\sigma^2$ is the noise variance per dimension and Poltyrev threshold $\sigma_{max}^2$ is given by
\begin{equation}
\label{equ_poltyrev_threshold}
\sigma_{max}^2=\frac{|\det(\mathbf{G})|^{\frac{2}{n}}}{2\pi e}.
\end{equation}

An integer-check matrix of a lattice is the inverse of a generator matrix, $\mathbf{H} = [\mathbf{h}_{ij}] = \mathbf{G}^{-1}$. 
The $n$ integer-check equations for a lattice point $\bs{x}$ are $\sum_{j} \; \mathbf{h}_{ij} \; x_j \in \Z$, for $i=1 \ldots n$.  
\textit{Low-density lattice codes} (LDLC) are a special class of lattices proposed by Sommer {\it et al.} in \cite{Sommer2008}.
A LDLC is defined by a sparse integer-check matrix in order to allow for iterative decoding in high dimensions. We now define the Latin square LDLC considered in this paper.
\begin{definition}[Latin Square LDLC \cite{Sommer2008}]
 A LDLC is called {\em Latin square LDLC} if all the row degrees and column degrees of its integer-check matrix are equal to a common degree $d$ and if every row and column of the integer-check matrix has the same nonzero values, except for a possible change of order and random signs. The sorted sequence of these $d$ values $h_1 \ge h_2 \ge \cdots \ge h_d > 0$ is referred to as the generating sequence of the Latin square LDLC.
\end{definition}

Other families of lattices built from sparse codes have been recently proposed, 
such as lattices from Construction A with non-binary low-density parity-check codes \cite{diPietro2012}\cite{diPietro2013}, 
referred to as LDA lattices. More recently, a powerful family of {\it generalized low-density}  (GLD) lattices for the AWGN channel has been defined by the intersection of repeated and interleaved lattices \cite{Boutros2014}. 
Standard LDLC, LDA, and GLD lattices exhibit no diversity and should be modified as in \refeq{eq:h_fl_2} in presence of block fading.
%

\subsection{Non-Ergodic Fading Channel Model} \label{sec_channel_model}
We assume coherent detection with perfect channel state information at the receiver side only.
The fading channel is flat, i.e., there are no multiple paths \cite{Proakis2008}. 
Fading coefficients are real non-negative with a Rayleigh distribution. 
If $\alpha$ denotes a fading coefficient, then $p(\alpha)=2\alpha \exp(-\alpha^2)$ 
or equivalently $p(\alpha^2)=\exp(-\alpha^2)$, for $0 \le \alpha < +\infty$.
It is worth noting that results in this paper do not rely on this particular distribution of fading, 
they are still valid for most usual fading distributions, e.g. the Nakagami distribution of order $m$.

Let $\bs{x}$ be a lattice point. Consider the non-ergodic fading 
where fading coefficients take only $L$ values within a lattice point, $2 \le L < n$.
The non-ergodic BF channel with diversity $L$ has the following mathematical model: 
\begin{equation}
\label{equ_rayleigh_any_L}
y_i=\alpha_j x_i + \eta_i, ~~~j= \left\lceil \frac{i}{n/L} \right\rceil, ~~~i=1,2,\ldots n,
\end{equation}
\noindent
where $\alpha_j$ are independent and identically distributed (i.i.d.) 
Rayleigh distributed fading coefficients and $\eta_i \sim \N(t; 0, \sigma^2)$.

Let $\gamma$ be the signal-to-noise ratio (SNR) for an infinite lattice constellation,
\begin{equation}
\gamma=\frac{|\det(\mathbf{G})|^{\frac{2}{n}}}{\sigma^2}. \label{equ_snr}
\end{equation}
Assuming maximum-likelihood (ML) decoding at the receiver side, let $P_e(\Lambda)$ be the point error probability.
On the BF channel defined by (\ref{equ_rayleigh_any_L}), the diversity order of $P_e(\Lambda)$ is defined
by its slope at high SNR \cite{Proakis2008}
\begin{equation}
\lim_{\gamma \rightarrow \infty} \frac{-\log(P_e(\Lambda))}{\log(\gamma)}
\end{equation}
If $f$ and $g$ have the same diversity order then we denote this by $f(\gamma) \doteq g(\gamma)$,
i.e., $\frac{-\log(f)}{\log(\gamma)}= \frac{-\log(g)}{\log(\gamma)}$ for $\gamma \rightarrow \infty$. 
Similarly we introduce the notation with inequalities $\dot{\le}$ and $\dot{\ge}\;$.

\begin{definition}[Full-diversity Lattice under ML Decoding]
Consider a BF channel with $L$ independent fading coefficients per lattice point.
$\Lambda$ is a full-diversity lattice under ML decoding if $P_e(\Lambda) \doteq \frac{1}{\gamma^L}$.
\end{definition}

%

%
\section{Poltyrev Outage Limit For Infinite Lattice Constellations}\label{sec:POL}
%
As mentioned earlier, no attempt has been made in the literature to derive
a limit for the BF channel equivalent to that of Poltyrev for the Gaussian channel. In this section, we derive a outage limit, referred to as {\it Poltyrev outage limit}, for lattices over BF channel.

Let $\bs{\alpha} = \textrm{diag} \left( \alpha_{1}, \; . \; , \alpha_{1}, \alpha_{2}, \; .\; , \alpha_{2}, \ldots, \alpha_{L}, \; . \; , \alpha_{L} \right)$ be the $n \times n$ diagonal matrix including the $L$ fading coefficients, 
each repeated $n/L$ times as defined by the model in (\ref{equ_rayleigh_any_L}).
Given the lattice generator matrix $\mathbf{G}$, after going through the BF channel, the new lattice
added to the AWG noise has the generator matrix $\mathbf{G}_{\textrm{BF}}=\bs{\alpha} \mathbf{G}$. 
The fundamental volume becomes
\begin{equation}
\label{eq:det_new_G}
\left|\det\left(\mathbf{G}_{\textrm{BF}}\right)\right| = |\det\left(\mathbf{G}\right)| \times  \prod_{l = 1}^{L} \alpha_l^{n/L}. 
\end{equation}

For a fixed instantaneous fading $\bs{\alpha}$, after combining (\ref{equ_poltyrev_threshold}) 
and~(\ref{eq:det_new_G}), Poltyrev threshold becomes
\begin{equation}
\sigma_{max}^2(\bs{\alpha})=\frac{\prod_{l = 1}^{L} \alpha_l^{2/L} |\det(\mathbf{G})|^{\frac{2}{n}}}{2\pi e}.
\end{equation}
Decoding of the infinite lattice constellation is possible with a vanishing error probability 
if $\sigma^2 < \sigma_{max}^2(\bs{\alpha})$ \cite{Poltyrev1994}. Hence, for variable fading,
an outage event occurs whenever $\sigma^2~>~\sigma_{max}^2(\bs{\alpha})$. 
The Poltyrev outage limit $P_{\textrm{out}}(\gamma)$ is then defined by the following probability
\begin{align}
\label{equ_poltyrev_outage}
 P_{\textrm{out}}(\gamma) & =  
P\left( \sigma^2 > \frac{\prod_{l = 1}^{L} \alpha_l^{2/L} |\det(\mathbf{G})|^{\frac{2}{n}} }{2\pi e} \right) \nonumber \\
& =  P\left(\prod_{l = 1}^{L} \alpha_l^2 < \frac{(2\pi e)^L}{\gamma^L} \right). 
\end{align}

$P_{\textrm{out}}(\gamma)$ does not admit a closed-form expression but it can be numerically estimated via the Monte Carlo method.
The point error rate after lattice decoding, for a given lattice over a BF channel, can be compared
to $P_{\textrm{out}}(\gamma)$ to validate the diversity order and the gap in signal-to-noise ratio. 
But most importantly, the equality $\prod_{l = 1}^{L} \alpha_l^2 =\frac{(2\pi e)^L}{\gamma^L}$ 
defines a boundary in the fading space below which outage events occur.
This boundary shall be called Poltyrev outage boundary. 
In the next section we prove that $P_{\textrm{out}}(\gamma) \doteq \frac{1}{\gamma^{L}}$.
%
\section{Diversity Order of Poltyrev Outage Limit}
It is well known from Maximum Ratio Combining techniques on fading channels \cite{Proakis2008} that 
$P\left( \sum_{l = 1}^{L} \alpha_l^2 < \frac{1}{\gamma} \right)$ has diversity $L$. 
Is it true for $\prod_{l} \alpha_l^2$? Fortunately, the expression of Poltyrev outage limit in (\ref{equ_poltyrev_outage})
compares the product of squared fadings to $\frac{1}{\gamma^{L}}$, not to $\frac{1}{\gamma}$. 
Roughly speaking, $P_{\textrm{out}}(\gamma)$ behaves like $\left[P\left(\alpha_1^2 < \frac{1}{\gamma}\right)\right]^L$ leading to diversity $L$.
Let us discuss the exact proof.

The proof of $P_{\textrm{out}}(\gamma) \doteq \frac{1}{\gamma^{L}}$ is made by induction; 
first for the special case of $L=2$ and later for an arbitrary value of $L$. 
The constant term $2\pi e$ can be embedded into $\gamma$, so we have
\begin{align}
P_{\textrm{out}}(\gamma) \doteq P\left(\gamma^{L} \prod_{l = 1}^{L} \alpha_l^2 < 1 \right). \label{eq:error_rate_prop}
\end{align}
The equality in diversity is reached after proving a lower bound and an upper bound for $P_{\textrm{out}}(\gamma)$.
In other words,  $P_{\textrm{out}}(\gamma) \dot{\ge} \frac{1}{\gamma^{L}}$ and $P_{\textrm{out}}(\gamma) \dot{\le} \frac{1}{\gamma^{L}}$
is equivalent to $P_{\textrm{out}}(\gamma) \doteq \frac{1}{\gamma^{L}}$. The reader should be aware
that a lower bound on the error probability yields an upper bound on diversity, and vice versa.

\begin{lemma} \label{lm_pol_L_2}
Consider a BF channel with diversity $L=2$. The Poltyrev outage limit defined by (\ref{equ_poltyrev_outage}) satisfies
$P_{\textrm{out}}(\gamma) \doteq \frac{1}{\gamma^{2}}$.
\end{lemma}
\begin{proof} 
We now prove $P_{\textrm{out}}(\gamma) \doteq P \left(\alpha_1^2 \alpha_2^2\; \gamma^{2} < 1 \right) \doteq \frac{1}{\gamma^{2}}$.
Recall that a Rayleigh distributed random variable $\alpha_i$ satisfies, for any $k>0$,
\begin{align}
P\left( \alpha_i^2 < \frac{1}{\gamma^k}\right) \doteq \frac{1}{\gamma^k} \; .\label{eq:ray_alpha_2}
\end{align}
%
{\em 1) Lower Bound:}
\begin{align*}
 P \left(\alpha_1^2 \alpha_2^2\; \gamma^{2} < 1 \right) &= P \left(\alpha_1^2 \alpha_2^2\; \gamma^{2} < 1 \; | \; \alpha_2^2 < 1\right) P \left( \alpha_2^2 < 1 \right) \\
 &+ P \left(\alpha_1^2 \alpha_2^2\; \gamma^{2} < 1 \; | \; \alpha_2^2 > 1\right) P \left( \alpha_2^2 > 1 \right) \; .
\end{align*}
 We lower bound by the first term to get
\begin{align*}
 &P\left(\alpha_1^2 \alpha_2^2\; \gamma^{2} < 1 \right) \ge P \left(\alpha_1^2 \alpha_2^2\; \gamma^{2} < 1 \; | \; \alpha_2^2 < 1\right) P \left( \alpha_2^2 < 1 \right) \\
&\ge P \left(\alpha_1^2 \; \gamma^{2} < 1 \right) P \left( \alpha_2^2 < 1 \right) 
\doteq P \left(\alpha_1^2 \; \gamma^{2} < 1 \right),
\end{align*}
which gives
\begin{align}
&P \left(\alpha_1^2 \alpha_2^2\; \gamma^{2} < 1 \right) \dot{\ge} \frac{1}{\gamma^2} \; .\label{eq:POL_L_2_LB}
\end{align}
%
%
\noindent
2) Upper Bound: To derive an upper bound on $P\left(\alpha_1^2 \alpha_2^2\; \gamma^{2} < 1 \right)$, 
we plot the boundary $\alpha_1^2 \alpha_2^2 = \frac{1}{\gamma^{2}}$ which is a hyperbola as shown in blue in \reffig{pic:p2}
(to be taken for $L=2$).
We partition the total area under the outage boundary into three regions: 
a) $\mathcal{R}$ is the area where $\alpha_1^2 < \frac{1}{\gamma}$ and $\alpha_2^2 < \frac{1}{\gamma}$; 
b) $\mathcal{T}_1$ is the area where $\alpha_1^2 < \frac{1}{\gamma}$ and $\alpha_2^2 > \frac{1}{\gamma}$ 
but $\alpha_1^2 \alpha_2^2 < \frac{1}{\gamma^{2}}$; and 
c) $\mathcal{T}_2$ is the area where $\alpha_1^2 > \frac{1}{\gamma}$ and $\alpha_2^2 < \frac{1}{\gamma}$ 
but $\alpha_1^2 \alpha_2^2 < \frac{1}{\gamma^{2}}$. 
The areas $\mathcal{T}_1$ and $\mathcal{T}_2$ are equal (for $L=2$). 
So we can write
\begin{align}
 P \left(\alpha_1^2 \alpha_2^2\; \gamma^{2} < 1 \right) 
 &= P(\mathcal{R}) + 2 P(\mathcal{T}_2) \; .\label{eq:ub}
\end{align}
%
%
%
%
Then we have
\begin{align}
 &P(\mathcal{R}) = P \left(\alpha_1^2 \in \left[0,\frac{1}{\gamma} \right] \; \textrm{and} \; \alpha_2^2 \in \left[0,\frac{1}{\gamma} \right]\right) \nonumber \\
%
%
 &= \left( P \left(\alpha_1^2 \in \left[0,\frac{1}{\gamma} \right] \right) \right)^2 \; 
%
  = \left( P \left(\alpha_1^2 \le \frac{1}{\gamma} \right) \right)^2 
 \doteq \frac{1}{\gamma^2} . \label{eq:p_r}
\end{align}
%
%
%
%
%
\noindent
We now introduce $\phi(x) = 1 - e^{-x}$ which is required in the sequel. 
For $x \ge 0$, it can be shown that
\begin{align}
\phi(x) = 1 - e^{-x} \le \min(1, x)\; . \label{eq:f_e_x}
\end{align}

\noindent
Now denote $X = \alpha_1^2$ and $Y = \alpha_2^2$. Then
\begin{equation}
\label{eq:p_t_1}
P(\mathcal{T}_2) = \int_{\mathcal{T}_2} p_{X,Y}(x,y) = \int_{x=\frac{1}{\gamma}}^{\infty} \int_{y=0}^{\frac{1}{x\gamma^2}} e^{-x} e^{-y} dx \; dy,
\end{equation}
after integrating over $y$, we get
\begin{align}
 P(\mathcal{T}_2) &= \int_{x=\frac{1}{\gamma}}^{\infty} \left(1 - e^{-\frac{1}{x\gamma^2}} \right) e^{-x} \; dx.\label{eq:T}
\end{align}

\noindent
The previous integral has $\frac{1}{x\gamma^2} \le \frac{1}{\gamma} \le 1$ at high SNR. We get
\begin{align*}
\phi\left(\frac{1}{x\gamma^2} \right) \le \min\left( 1, \frac{1}{x\gamma^2}\right) = \frac{1}{x\gamma^2} \;.
\end{align*}
With this result, \refeq{eq:T} can be upperbounded as
\begin{align}
P(\mathcal{T}_2) &\le \int_{x=\frac{1}{\gamma}}^{\infty} \frac{1}{x\gamma^2} e^{-x} \; dx = \frac{\Delta}{\gamma^2} \; ,\label{eq:T_2} 
\end{align}
where $\Delta$ is defined as
\begin{align}
 \Delta = \int_{x=\frac{1}{\gamma}}^{\infty} \frac{e^{-x} }{x} \; dx. \label{eq:delta}
\end{align}
Solving \refeq{eq:delta} using integration by parts yields
\begin{align}
 \Delta = e^{-\frac{1}{\gamma}} \ln(\gamma) + \int_{\frac{1}{\gamma}}^{\infty} \ln(x) e^{-x} dx .\label{eq:delta_2}
\end{align}
With $e^{-\frac{1}{\gamma}} \le 1$ and $\ln(x) \le x$, \refeq{eq:delta_2} can be upperbounded as
\begin{align}
\Delta \le \ln(\gamma) + \frac{1}{\gamma}e^{-\frac{1}{\gamma}} + e^{-\frac{1}{\gamma}}  \label{eq:delta_3}
\end{align}
\noindent
Substituting the upper bound on $\Delta$ in \refeq{eq:T_2}, 
\begin{align}
 &P(\mathcal{T}_2) \le \frac{ \ln(\gamma) + \frac{1}{\gamma}e^{-\frac{1}{\gamma}} + e^{-\frac{1}{\gamma}} } { \gamma^2 } \; 
 \Rightarrow \; 
 P(\mathcal{T}_2) \; \dot{\le} \; \frac{ 1 } { \gamma^2 } \label{eq:p_t_2}
\end{align}
\noindent
Using \refeq{eq:p_r} and \refeq{eq:p_t_2} in \refeq{eq:ub}, we have 
\begin{align}
P \left(\alpha_1^2 \alpha_2^2\; \gamma^{2} < 1 \right) \; \dot{\le} \;  \frac{ 1 } { \gamma^2 } \label{eq:POL_L_2_UB}
\end{align}
which is the desired upper bound.\\
\noindent
It can be concluded from the upper bound and lower bound as given in \refeq{eq:POL_L_2_LB} and \refeq{eq:POL_L_2_UB} that
\[
P \left(\alpha_1^2 \alpha_2^2\; \gamma^{2} < 1 \right) \doteq \frac{1}{\gamma^{2}} \;.
\qedhere
\]
\end{proof}
%
%
%
%
%
Now the previous lemma is generalized by induction to an arbitrary value of diversity order $L \ge 2$.
\begin{theorem} \label{thm_pol_L}
Consider a BF channel with diversity $L \ge 2$. 
The Poltyrev outage limit defined by (\ref{equ_poltyrev_outage}) satisfies $P_{\textrm{out}}(\gamma) \doteq \frac{1}{\gamma^{L}}$.
\end{theorem}
\begin{proof}
Let us assume that the theorem statement is true for $L-1$, i.e.
\begin{align}
P \left(\prod_{i=1}^{L-1} \alpha_i^2 \; \gamma^{(L-1)} < 1 \right) \doteq \frac{1}{\gamma^{(L-1)}} \; .\label{eq:P_Y_L-1}
\end{align} 
Now, let us prove it for a diversity order $L$.
As for \reflm{lm_pol_L_2}, we derive upper and lower bounds for $P_{\textrm{out}}(\gamma)$.\\~\\
%
{\em 1) Lower Bound:} In a way similar to the proof of the lower bound in \reflm{lm_pol_L_2},
the term with $\alpha_L^2 \gamma > 1$ is dropped, so we have 
\begin{align*}
 P_{\textrm{out}}(\gamma) & \doteq P\left(\prod_{i=1}^{L-1} \alpha_i^2 \; \gamma^{(L-1)} \; \alpha_L^2 \gamma < 1 \right) \\
 &\ge P\left(\prod_{i=1}^{L-1} \alpha_i^2 \; \gamma^{(L-1)} < 1 \right) P\left( \alpha_L^2 \gamma < 1 \right) \\
 &\dot{\ge} ~\frac{1}{\gamma^{(L-1)} \gamma} \;\; = \;\; \frac{1}{\gamma^{L}} \;,
\end{align*}
which gives the required lower bound with diversity $L$. The upper bound
relies on the partitioning of the area under the outage boundary.

\noindent
{\em 2) Upper Bound:} We partition the area under the outage boundary into three regions: 
a) $\mathcal{R}$ is the area where $\prod_{i=1}^{L-1} \alpha_i^2 < \frac{1}{\gamma^{L-1}}$ and $\alpha_L^2 < \frac{1}{\gamma}$; 
b) $\mathcal{T}_1$ is the area where $\prod_{i=1}^{L-1} \alpha_i^2 < \frac{1}{\gamma^{L-1}}$  and $\alpha_L^2 > \frac{1}{\gamma}$ 
but $\prod_{i=1}^{L} \alpha_i^2 < \frac{1}{\gamma^{L}}$; and 
c) $\mathcal{T}_2$ is the area where $\prod_{i=1}^{L-1} \alpha_i^2 > \frac{1}{\gamma^{L-1}}$ and $\alpha_L^2 < \frac{1}{\gamma}$ 
but $\prod_{i=1}^{L} \alpha_i^2 < \frac{1}{\gamma^{L}}$. 
The areas $\mathcal{T}_1$ and $\mathcal{T}_2$ are not equal for $L>2$. 
With this, we get
\begin{align}
P_{\textrm{out}}(\gamma) \doteq P(\mathcal{R}) + P(\mathcal{T}_1) + P(\mathcal{T}_2).\label{eq:P_Y}
\end{align}
\begin{align}
P(\mathcal{R}) &= P \left(\alpha_{L}^2 \in \left[0,\frac{1}{\gamma} \right] \; \textrm{and} \; \alpha_1^2\cdots\alpha_{L-1}^2 \in \left[0,\frac{1}{\gamma^{(L-1)}} \right]\right) \nonumber \\
 &= P \left(\alpha_L^2 \le \frac{1}{\gamma} \right)  P \left(\alpha_1^2\cdots\alpha_{L-1}^2 \le \frac{1}{\gamma^{(L-1)}} \right) \nonumber \\
&\doteq \frac{1}{\gamma^L}. \label{eq:p_r_2}
\end{align}
The last equality is derived from \refeq{eq:ray_alpha_2} and \refeq{eq:P_Y_L-1}.
For the calculation of $P(\mathcal{T}_1)$ and $P(\mathcal{T}_2)$, 
let $X = \prod_{i=1}^{(L-1)} \alpha_i^2$ and let $Y = \alpha_L^2$.\\
%
%
\begin{figure}[!t]
\begin{center}
\includegraphics[width=0.35\textwidth]{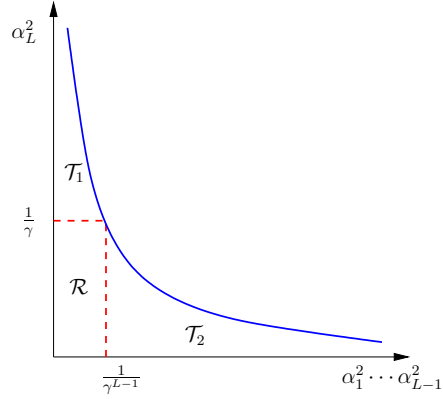}
\end{center}
\caption{Boundary in the fading plane defined by a constant product of squared fadings at a given SNR,
$\alpha_1^2\alpha_2^2 \cdots \alpha_L^2 = \frac{1}{\gamma^L}$. The area under the Poltyrev outage boundary 
is partitioned into three regions $\mathcal{R}$, $\mathcal{T}_1$, and $\mathcal{T}_2$.\label{pic:p2}}
\end{figure}
%
%
%
%
%
Calculation of $P(\mathcal{T}_1)$:\\
\begin{equation}
P(\mathcal{T}_1) = \int_{y=\frac{1}{\gamma}}^{\infty} e^{-y} \; dy \int_{x=0}^{\frac{1}{y\gamma^{L}}} p_X(x) \; dx	\label{eq:T_2_L_1}
\end{equation}
From \refeq{eq:P_Y_L-1} we get,
\begin{align*}
\int_{x=0}^{\frac{1}{y\gamma^{L}}} p_X(x) \; dx = P \left(\prod_{i=1}^{L-1} \alpha_i^2 \; < \frac{1}{y\gamma^{L} } \right) \doteq \frac{1}{y\gamma^L}.
\end{align*}
\noindent
\refeq{eq:T_2_L_1} can be rewritten as
\begin{align}
P(\mathcal{T}_1) \doteq \frac{1}{\gamma^{L}} \int_{y=\frac{1}{\gamma}}^{\infty} \frac{e^{-y}}{y} \; dy.   \label{eq:T_2_L_2}
\end{align}
\noindent
Since $\frac{1}{\gamma} \to 0$, the exponential integral is given by \cite{Bender1978},
\begin{align*}
\int_{y=\frac{1}{\gamma}}^{\infty} \frac{e^{-y}}{y} \; dy = 0.5772 + \ln(\gamma) + o\left(\frac{1}{\gamma}\right).
\end{align*}
Then, we have for the area $\mathcal{T}_1$
\begin{equation}
\label{eq:P_T1_L}
P(\mathcal{T}_1) \doteq \frac{\ln(\gamma) }{\gamma^{L}} \doteq \frac{1}{\gamma^{L}}. 
\end{equation}
%
%
%
%
%
%
Calculation of $P(\mathcal{T}_2)$:\\
\begin{align*}
P(\mathcal{T}_2) = \int_{y=0}^{\frac{1}{\gamma}} e^{-y} \; dy \int_{x=\frac{1}{\gamma^{(L-1)}}}^{\frac{1}{y\gamma^{L}}} p_X(x) \; dx.
\end{align*}
\noindent
$\forall \epsilon > 0 $, let $L_0 = L - \epsilon$ and $y_0 = \frac{1}{\gamma^{L_0}}$ then $P(\mathcal{T}_2)$ is given by
\begin{align}
 P(\mathcal{T}_2) &= \int_{y=0}^{y_0} e^{-y} \; dy \int_{x=\frac{1}{\gamma^{(L-1)}}}^{\frac{1}{y\gamma^{L}}} p_X(x) \; dx \nonumber \\
  & + \int_{y_0}^{\frac{1}{\gamma}} e^{-y} \; dy \int_{x=\frac{1}{\gamma^{(L-1)}}}^{\frac{1}{y\gamma^{L}}} p_X(x) \; dx  \label{eq:T_1_L_1}
\end{align}
The upper limit of the inner integral goes to 0 because
\begin{align*}
\frac{1}{y\gamma^L} \le \frac{1}{y_0\gamma^L} = \frac{1}{\gamma^{(L-L0)}} = \frac{1}{\gamma^{\epsilon}}.
\end{align*}
\noindent
In the first term of (\ref{eq:T_1_L_1}), the inner integral is upperbounded by~1.
In the second term of (\ref{eq:T_1_L_1}), the inner integral is found by applying (\ref{eq:P_Y_L-1}) twice.
We reach an upper bound for $P(\mathcal{T}_2)$,
\begin{align}
P(\mathcal{T}_2) \le T_2 \label{eq:T_1_1}
\end{align}
\noindent
where
\begin{align}
 &T_2 \doteq \int_{y=0}^{y_0} e^{-y} dy 
+ \int_{y_0}^{\frac{1}{\gamma}} e^{-y} \left[ \frac{1}{y\gamma^{L}} - \frac{1}{\gamma^{L-1}}\right] dx \nonumber \\
 &\le y_0 + \frac{1}{\gamma^L} \int_{y_0}^{\frac{1}{\gamma}} \frac{e^{-y}}{y} \; dy 
%
 \;\; \dot{\le} \;\; \frac{1}{\gamma^{L - \epsilon}} + \frac{1}{\gamma^L} \int_{y_0}^{\frac{1}{\gamma}} \frac{e^{-y}}{y} \; dy\;. \label{eq:T_1_2}
\end{align}
\noindent
The evaluation of the exponential integral in \refeq{eq:T_1_2} gives \cite{Bender1978}
\begin{equation}
\int_{y_0}^{\frac{1}{\gamma}} \frac{e^{-y}}{y} \; dy 
= (L_0 -1) \ln(\gamma) + o\left( \frac{1}{\gamma} \right)\;. \label{eq:T_1_3}
\end{equation}
\noindent
Then, from \refeq{eq:T_1_1}, \refeq{eq:T_1_2}, and \refeq{eq:T_1_3}, $\forall \epsilon > 0$
\begin{equation} 
\label{eq:P_T2_L}
P(\mathcal{T}_2) \; \dot{\le} \; \frac{1}{\gamma^{L - \epsilon}} 
+ \frac{(L - \epsilon - 1)\ln(\gamma)}{\gamma^{L}} \doteq \frac{1}{\gamma^{L - \epsilon}} 
\end{equation}
%
%
%
%
Using \refeq{eq:p_r_2}, \refeq{eq:P_T1_L}, and \refeq{eq:P_T2_L}, 
we get the upper bound 
\begin{align} 
P_{\textrm{out}}(\gamma) \; \dot{\le} \;  \frac{1}{\gamma^L} \; .
\end{align}
\noindent
Similar to the case of $L=2$, we conclude from the lower bound and the upper bound derived above that 
\[
P_{\textrm{out}}(\gamma)  \; \doteq \;  \frac{1}{\gamma^L} \; . \qedhere
\]
%
\end{proof}
%
%
%
\section{Full-Diversity Construction of LDLC under ML decoding}
\noindent
Let $\Lambda$ be a real lattice of rank $n$ defined by a $n \times n$ 
integer-check matrix $H$. Assume that $n$ is multiple of $L$,
where $L$ is the diversity order of the block-fading channel.
Let us write $H$ in the form
\begin{equation}
H=\left[ ~ \tilde{H}_1 ~|~ \tilde{H}_2 ~|~ \ldots ~|~ \tilde{H}_L ~\right],
\end{equation}
where $\tilde{H}_j$ is a $n \times n/L$ matrix, $j=1, \ldots, L$.
In the above expression of the integer-check matrix~$H$, the channel is assumed to have 
the same fading value $\alpha_j$ affecting all $n/L$ 
lattice components associated to the columns of $\tilde{H}_j$,
as defined in (\ref{equ_rayleigh_any_L}) in \refsec{sec_channel_model}. 
Using the $L$ submatrices $\tilde{H}_j$, let us build
a new shortened integer-check matrix $\Psi_k$ of size $n \times \ell n/L$
by combining $\ell$ submatrices out of $L$,
for $\ell=1, \ldots, L-1$. %
The number of shortened integer-check matrices is $K = \sum_{\ell=1}^{L-1} = 2^L-2$ since the empty matrix $\Psi_0$ and the full matrix $\Psi_{2^L-1} = H$ correspond to two trivial cases.
Also, for $k=1, \ldots, K$, we define the function $\kappa(k)$ as 
\begin{equation}
\kappa(k)=i\times \frac{n}{L}
\end{equation}
for $i$ that satisfies
\begin{equation}
\sum_{\ell=1}^{i-1} {L \choose \ell} \le k \le \sum_{\ell=1}^{i} {L \choose \ell},
\end{equation}
such that $\Psi_k$ is a $n \times \kappa(k)$ matrix.

For example, for $L=2$, we have $\Psi_1=\tilde{H}_1$, $\Psi_2=\tilde{H}_2$,
and $\kappa(1)=\kappa(2)=n/L$. 
For $L=3$, we have $\Psi_1=\tilde{H}_1$, $\Psi_2=\tilde{H}_2$, $\Psi_3=\tilde{H}_3$,
$\Psi_4=[\tilde{H}_1|\tilde{H}_2]$, $\Psi_5=[\tilde{H}_1|\tilde{H}_3]$,
$\Psi_6=[\tilde{H}_2|\tilde{H}_3]$, $\kappa(1)=\kappa(2)=\kappa(3)=n/L$,
and $\kappa(4)=\kappa(5)=\kappa(6)=2n/L$.

One of the purpose of this paper is to build low-density lattices 
such that $P_e \approx K_e/\gamma^L$ 
at large signal-to-noise ratio, where $K_e$ is a non-negative real constant. 
The so called coding gain~\cite{Tse2005} is given by $1/\sqrt[L]{K_e}$.
Maximizing the coding gain or equivalently minimizing $K_e$ is not the subject
of this paper.
Following the above definition and using the $K$ shortened integer-check matrices, 
we can now state a simple lemma useful for proving full diversity.

\begin{lemma} \label{lm_fulldiv_ml} 
A lattice $\Lambda$ is full-diversity under ML decoding on an L-diversity
block-fading channel if and only if $\Psi_k x \in \Z^n$ admits $x=0 \in \R^{\kappa(k)}$ as unique solution, $\forall k=1, \ldots, 2^L-2$.
\end{lemma}

\begin{proof} It is assumed that the 0 point has been transmitted
since point error probability does not depend on the transmitted point.
The key idea is to prove full diversity in presence of block erasures,
an approach recently used in \cite{Boutros2009}\cite{Boutros2009ita}.\\
\noindent
Let us start with the case $L=2$. For $x \in \Lambda$, 
let $P(0 \rightarrow x)$ denote
the pairwise error probability, i.e., the probability of $\|y-\alpha\cdot x\|^2 < \|y-\alpha\cdot 0\|^2$
when point 0 is transmitted. Using standard tools from Communications Theory \cite{Tse2005},
it is easy to prove a well-known upperbound for the pairwise error probability
on a block fading channel with double diversity as defined in \refeq{equ_rayleigh_any_L},
\begin{equation}
\label{equ_pairwise}
P(0 \rightarrow x) \le \frac{1}{1+\frac{\gamma}{8}(\sum_{i=1}^{n/2} x_i^2)} \times 
\frac{1}{1+\frac{\gamma}{8}(\sum_{i=n/2+1}^n x_i^2)}
\end{equation}
where $\gamma$ is defined by \refeq{equ_snr}. The above bound is known to be very loose,
but it is sufficient for our study. Consider the following statements:
\begin{itemize}
\item $\mathcal{S}_1$: ~$x_{n/2+1}=x_{n/2+2}=\ldots=x_n=0$~~~$\Rightarrow$~~~$x_1=x_2=\ldots=x_{n/2}=0$.
\item $\mathcal{S}_2$: ~$x_1=x_2=\ldots=x_{n/2}=0$~~~$\Rightarrow$~~~$x_{n/2+1}=x_{n/2+2}=\ldots=x_n=0$.
\item $\mathcal{S}_3= (\mathcal{S}_1~\text{and}~\mathcal{S}_2$).
\end{itemize}
From (\ref{equ_pairwise}) and the fact that $P_e(\Lambda) \le \sum_{x \in \Lambda-\{0\}} P(0 \rightarrow x)$, 
we see that $\mathcal{S}_3$ is true if and only if $P_e(\Lambda)$ decreases as $1/\gamma^2$, i.e.,
full diversity is attained. On the other hand, $\mathcal{S}_1$ is equivalent to $\Psi_1 x \in \Z^n$ admitting $0 \in \R^{n/2}$ as a unique
solution and $\mathcal{S}_2$ is equivalent to $\Psi_2 x \in \Z^n$ admitting $0 \in \R^{n/2}$ as a unique
solution,  this completes the proof for $L=2$.\\
\noindent
In order to easily extend to $L>2$, let us interpret $\mathcal{S}_1$ and $\mathcal{S}_2$.
Statement $\mathcal{S}_1$ corresponds to a context where $n/2$ components $(x_{n/2+1},x_{n/2+2}, \ldots, x_n)$
are perfectly known, i.e., $\gamma=+\infty$ and $\alpha_2=1$ (or any non-zero value).
Also, it is assumed that the $n/2$ components $(x_1, x_2, \ldots, x_{n/2})$ are unknown (erased)
and the constraint $\Psi_1 x \in \Z^n$ gives all possible solutions. Thus,
the above proof is based on a block erasure channel.
The lattice $\Lambda$ is full-diversity under ML decoding iff $n/2$ erased components
$(x_1, x_2, \ldots, x_{n/2})$ can be uniquely determined from perfectly known $n/2$ components 
$(x_{n/2+1},x_{n/2+2}, \ldots, x_n)$ and vice versa.\\
\noindent
Now, for any $L \ge 2$, we can state that $\Lambda$ is full-diversity on the fading channel
defined by (\ref{equ_rayleigh_any_L}) iff $\ell$ erased blocks of components can be uniquely
determined from the perfectly known (non-erased) lattice components, where $1 \le \ell \le L-1$.
This is equivalent to $0$ being the unique solution of the $2^L-2$ constraints $\Psi_k x \in \Z^n$. 
\end{proof}

Notice that \reflm{lm_fulldiv_ml} does not need $H$ to be a low-density matrix.
It is valid for both sparse and non-sparse matrices.
It is possible to simplify \reflm{lm_fulldiv_ml} by reducing the number of 
constraints to $L$ instead of $2^L-2$, as stated in \reflm{lm_fulldiv_ml_2}. 
Nevertheless, we believe that \reflm{lm_fulldiv_ml} is more useful for
lattice construction than \reflm{lm_fulldiv_ml_2} because we usually
start constructing a lattice in a recursive way by imposing lower diversity orders
before reaching full diversity.
Consider the $L$ largest $\Psi_k$ integer-check matrices, those whose size is $n \times (n-n/L)$
and index is $k=2^L-L-1, \ldots, 2^L-2$.
Let us refer to these matrices by $\Theta_k=\Psi_{k+2^L-L-2}$, for $k=1, \ldots, L$.

\begin{lemma} \label{lm_fulldiv_ml_2} 
A lattice $\Lambda$ is full-diversity under ML decoding on an L-diversity 
block-fading channel if and only if 
$\Theta_k x \in \Z^n$ admits $x=0 \in \R^{n-n/L}$ as the unique solution, $\forall k=1, \ldots, L$.
\end{lemma}

Proof for \reflm{lm_fulldiv_ml_2} is based on \reflm{lm_fulldiv_ml} and omitted here.

As a direct application, \reflm{lm_fulldiv_ml} is followed by \refthm{thm_fulldiv_ml} stating how to construct a full-diversity lattice under ML decoding for $L=2$. It is straightforward to generalize the proposed construction to $L>2$. In the rest of this paper, unless otherwise stated, we shall restrict the study
to $L=2$.

\begin{theorem} \label{thm_fulldiv_ml} 
Consider a double-diversity block-fading channel. Let $H=[h_{ij}]$ be the $n \times n$ integer-check
matrix of a real lattice $\Lambda$ of even rank $n$, where $h_{ij} \in \Q$, the field
of rationals. Let us decompose $H$ into four $n/2 \times n/2$ submatrices as follows
\begin{align}
H=\left[
\begin{array}{cc}
A & B \\
C & D 
\end{array}
\right]\label{eq:full_div_ml}.
\end{align}
\noindent
Assume that $A$, $B$, $C$, and $D$ have full rank. 
Let $\theta_1$ and $\theta_2$ be two algebraic numbers of degree $\ge 2$
such that $\theta_2/\theta_1 \notin \Q$. Then, the two lattices defined respectively by the
integer-check matrices
\begin{equation}
\left[
\begin{array}{cc}
\theta_1 A & \theta_1 B \\
\theta_2 C & \theta_2 D
\end{array}
\right]~~~\text{and}~~~
\left[
\begin{array}{cc}
\theta_1 A & \theta_2 B \\
\theta_2 C & \theta_1 D
\end{array}
\right]
\label{eq:h_fl_2}
\end{equation}
are full-diversity lattices under ML decoding.
\end{theorem}
\begin{proof}
Consider the first constraint $\Psi_1 x \in \Z^n$, where $x \in \R^{n/2}$ and
\[
\Psi_1=\left[ 
\begin{array}{c}
\theta_1 A \\
\theta_2 C 
\end{array}
\right].
\]
The upper half is $\theta_1 A x \in \Z^{n/2}$, we get $x \in \bar{\theta}_1 \Q^{n/2}$.
Similarly, the lower half is $\theta_2 C x \in \Z^{n/2}$, we get $x \in \bar{\theta}_2 \Q^{n/2}$.
But since $\theta_1 \Q \cap \theta_2 \Q = \{0\}$ we obtain $x=0$.\\
A similar reasoning can be made for $\Psi_2$. Then, applying \reflm{lm_fulldiv_ml} completes the proof of full diversity.
\end{proof}

\noindent
The weak condition $\theta_2/\theta_1 \notin \Q$ enables us to use conjugate
algebraic numbers from the same number field, e.g., take $\theta_1=\frac{1+\sqrt{5}}{2}$
and $\theta_2=\frac{1-\sqrt{5}}{2}$ in $\Q(\sqrt{5})$. A stronger condition
may be defined as $\Q(\theta_1) \cap \Q(\theta_2)=\Q$ and could be beneficial for the coding gain but it is not required for full-diversity. 

Another interesting method to construct full-diversity LDLC for ML decoding which is derived from \refthm{thm_fulldiv_ml} is given below. In this method, $\theta_1$  and $\theta_2$ are not explicitly used rather they are embedded within nonzero coefficients of the integer-check matrix such that the resulting LDLC is a Latin square LDLC
This construction method is given in following theorem. 

\begin{theorem} \label{thm_fulldiv_ml_Latin} 
Consider a double-diversity block-fading channel and let $H=[h_{ij}]$ be the $n \times n$ integer-check matrix of a real lattice $\Lambda$ of even rank $n$ and degree $d$. We decompose $H$ into four $n/2 \times n/2$ submatrices $A, B, C$ and $D$ as in \refeq{eq:full_div_ml} where each submatrix is full rank. 
Further assume that $A$ and $D$ are random regular matrices of degree $d - 1$, $B$ and $C$ are permutation matrices such that $H$ %
is an integer-check matrix of a Latin square LDLC with the generating sequence $\{h_1, h_2, \cdots, h_d \} = \left\{1, \theta, \cdots, \theta \right\}$ where $\theta = \frac{1}{\sqrt{d}} \text{ if } d \text{ is odd otherwise } \theta = \frac{1}{\sqrt{d+1}}$. 
Then, the LDLC defined by such an integer-check matrix is a full-diversity lattice under ML decoding.
\end{theorem}

The proof of this theorem is similar to the one in \refthm{thm_fulldiv_ml} and omitted here.
%
%
%
%
%
As shown in the next section, the same construction as \refthm{thm_fulldiv_ml} which combines $H$, $\theta_1$, and $\theta_2$, leads to a full-diversity lattice under iterative belief propagation decoding. A supplementary condition on the binary image of $H$ is required to accomplish full-diversity with iterative decoding.

\section{Full-Diversity Construction of LDLC under Iterative Decoding}
\noindent
We keep restricting the study to the default diversity order $L=2$ unless otherwise stated.
Hereafter, we consider only real LDLC under iterative decoding, 
i.e., real lattices $\Lambda$ with a sparse $n \times n$  integer-check matrix $H$.
The lattice constraint $Hx \in \Z^n$ admits a bipartite graph representation as follows:
(i) Draw $n$ vertices (variable nodes) on the left representing the $n$ lattice components $x_j$, $j=1, \ldots, n$.
(ii) Draw $n$ vertices (check nodes) on the right representing the $n$ rows $h_i$ of $H$ that define the $n$ LDLC constraints $h_i \cdot x = \sum_{j=1}^n h_{ij} x_j \in \Z$, $i=1, \ldots, n$. (iii) Link $x_j$ and $h_i$ by an edge if $h_{ij} \ne 0$. That edge has a multiplicative weight $h_{ij}$.

The factor graph \cite[Chap. 2]{Richardson2008} defined above is used for 
iterative belief propagation decoding of $\Lambda$ \cite{Sommer2008}. Usually, $H$ is regular
with $d$ nonzero entries per row and $d$ nonzero entries per column, $d \ll n$.
Let $H_b$ be the incidence matrix of the factor graph, i.e., $H_b=[b_{ij}]$ where
$b_{ij}=1$ if $h_{ij} \ne 0$ otherwise $b_{ij}=0$.

\begin{definition}[Binary Image of a Lattice]
The binary image $\mathcal{C}(\Lambda)$ of $\Lambda$ is a binary code defined by its
integer-check matrix $H_b$. 
\end{definition}
As a direct consequence, the binary image $\mathcal{C}(\Lambda)$ has dimension $0$ 
($0$-rate) and length $n$. 
In general, for regular and irregular LDLC, the degree distribution of the binary image
from an edge perspective \cite{Richardson2008} 
is given by $\lambda(t)=\sum_i \lambda_i t^{i-1}$ on the left
and $\rho(t)=\sum_j \rho_j t^{j-1}$ on the rigth, where $\lambda(1)=\rho(1)=1$ 
and $\sum_i \frac{\lambda_i}{i} = \sum_j \frac{\rho_j}{j}$. 
\begin{definition}[Erasure Channel Condition]
We say that $\Lambda$ satisfies the Erasure Channel (EC) condition
if the binary code $\mathcal{C}(\Lambda)$ achieves full diversity \cite{Boutros2009ita} 
after a finite or an infinite number of decoding iterations.
\end{definition}
The EC condition is a necessary condition (but not sufficient) for $\Lambda$ to
achieve full diversity. In a way similar to the study of full-diversity LDPC codes,
the so-called {\em root-LDPC} \cite{Boutros2009ita}\cite{Boutros2009}, we redefine
full-diversity under iterative belief propagation. The symbol error probability
referred to as $P_{es}$ is the error probability per lattice component:
\begin{definition}[Full-diversity Lattice under Iterative Decoding]
Consider a fading channel with $L$ independent fading coefficients per lattice point.
$\Lambda$ is a full-diversity lattice under iterative decoding if the 
symbol error probability $P_{es}$ at the iterative probabilistic decoder output 
is proportional~to~$\frac{1}{\gamma^L}$, for $\gamma \gg 1$.
\end{definition}

Before analyzing the construction of \refthm{thm_fulldiv_bp} under iterative decoding,
let us take a look at LDLC lattices with a random structure. For random lattices
and asymptotically large $n$, $\mathcal{C}(\Lambda)$ is an ensemble of 0-rate binary LDPC codes
with left and right degree distributions defined by the polynomials $\lambda(t)$
and $\rho(t)$ respectively. If the degree distribution is well chosen, a 0-rate ensemble
can achieve the capacity of an ergodic binary erasure channel (BEC) 
with erasure probability $\epsilon_0$,
for any $\epsilon_0 < 1$. When $\Lambda$ is transmitted on a block-fading channel
with diversity order $L$, the random 0-rate LDPC ensemble $\mathcal{C}(\Lambda)$
will observe an ergodic binary erasure channel whose parameter is
\begin{equation}
\epsilon_0=\frac{n-\frac{n}{L}}{n}=1-\frac{1}{L}.
\end{equation}
This value of $\epsilon_0$ is in accordance with the size of the largest integer-check
matrices $\Theta_k$ used in \reflm{lm_fulldiv_ml_2} under ML decoding.

The diversity population evolution (DPE) tracks the fraction of full-diversity bits
with the number of decoding iterations. The DPE renders a standard Density Evolution (DE) on the BEC,
where $\epsilon_0=1-\frac{1}{L}$, and
\begin{equation}
\label{equ_DPE_L}
\epsilon_{i+1}=\left(1-\frac{1}{L}\right) \lambda\big(1-\rho(1-\epsilon_i)\big).
\end{equation}
The necessary condition EC for full diversity is achieved if $\epsilon_i \rightarrow 0$
when $i \rightarrow +\infty$, $i$ being the decoding iteration number. From (\ref{equ_DPE_L})
it is easy to prove the following propositions:
\begin{proposition}\label{pro_ec_d}
Consider a regular random LDLC ensemble with degree $d \ge 2$.
For $L=2$, the EC condition for full-diversity is satisfied
if and only if $d\le 7$.
\end{proposition}
The above proposition tells us that the diversity tunnel is open
for all regular random LDLC ensembles when $2 \le d \le 7$. The tunnel is closed
for $d\ge 8$. 
\begin{proposition}\label{pro_fd_L}
Consider a regular random LDLC ensemble with degree $d=3$.
The EC condition for full-diversity is satisfied
if and only if $2 \le L \le 6$.
\end{proposition}
The above proposition tells us that a 3-regular random LDLC can never
be full-diversity for $L\ge 7$. 
%
%
%
Irregular random LDLC may be useful to increase the fraction of full-diversity lattice components
at the first decoding iterations and to increase the upper limit for achievable $d$ and $L$.
As an example, $\lambda(t)=0.418683\cdot t+0.162635 \cdot t^2+0.418683\cdot t^5$ 
and $\rho(t)=t^2$ has a better DPE tunnel than the fully 3-regular case.

Now, we can state an equivalent to \refthm{thm_fulldiv_ml} in the iterative decoding
context. 
The construction mentioned here is given for $L=2$ and any average weight $d \ge 2$.
%
%
\begin{theorem} \label{thm_fulldiv_bp} 
Consider a double-diversity block-fading channel. Let $H=[h_{ij}]$ be the $n \times n$ integer-check matrix of a real lattice $\Lambda$ of even rank $n$ and degree $d$, where $h_{ij} \in \Q$, the field of rationals. 
%
Let us decompose $H$ into four $n/2 \times n/2$ submatrices as follows
\begin{align}
H=\left[
\begin{array}{cc}
A & B \\
C & D
\end{array}
\right] \label{eq:fd_bp_H}.
\end{align}

\noindent
Assume that the binary image of $H$ has the following structure:
\[
H_b=\left[
\begin{array}{cc|cc}
\Pi_1 & 0     &  B_2   &  \Pi_4 \\
 B_1  & \Pi_2 &  \Pi_3 &  0     \\ \hline
 0    & \Pi_6 &  \Pi_7 &  B_4 \\
\Pi_5 &  B_3  &  0     &  \Pi_8 
\end{array}
\right].
\]
where $\Pi_k, k \in \{1,\ldots,8\}$ are permutation matrices and $B_k, k \in \{1,\ldots,4\}$ are regular random matrices with weight $d-2$.\\
Let $\theta_1$ and $\theta_2$ be two algebraic numbers of degree $\ge 2$
such that $\theta_2/\theta_1 \notin \Q$. Then, the two lattices defined respectively by the
integer-check matrices
\begin{equation}
\left[
\begin{array}{cc}
\theta_1 A & \theta_1 B \\
\theta_2 C & \theta_2 D
\end{array}
\right]~~~\text{and}~~~
\left[
\begin{array}{cc}
\theta_1 A & \theta_2 B \\
\theta_2 C & \theta_1 D
\end{array}
\right] \label{eqn_h_bp}
\end{equation}
are full-diversity lattices under iterative probabilistic decoding.
\end{theorem}
\begin{proof}
See \refapp{app_ldlc_iter}.
\end{proof}
Other full-diversity constructions may also exist but we described one of the simplest methods in \refthm{thm_fulldiv_bp}. It is interesting to note that the full-diversity LDLC constructed as per the method given in \refthm{thm_fulldiv_ml} also shows full-diversity property under iterative probabilistic decoding.

\begin{theorem} \label{thm_fulldiv_bp_L_2} 
A LDLC constructed as per the method proposed in \refthm{thm_fulldiv_ml} is also full-diversity under iterative probabilistic decoding.
\end{theorem}

This theorem can be proved in a similar manner as \refthm{thm_fulldiv_bp} and the proof is omitted.
%
%
%
%
The construction methods proposed in \refthm{thm_fulldiv_bp} and \refthm{thm_fulldiv_bp_L_2} can be extended to arbitrary values of channel diversity $L$. In the following, we propose a LDLC construction method valid for the diversity order $L=4$ which is an extension of the method given in \refthm{thm_fulldiv_bp}.
%
Before we discuss the theorem for $L=4$, the DPE for $L=4$ from (\ref{equ_DPE_L}) needs to be examined from which we state following proposition:
\begin{proposition}\label{pro_ec_d_L4}
Consider a regular random LDLC ensemble with degree $d \ge 2$.
For $L=4$, the EC condition for full-diversity is satisfied
if and only if $d\le 3$.
\end{proposition}

We can state following from \refpro{pro_ec_d_L4}: in case of $L=4$, the diversity tunnel is open only for regular random LDLC ensembles when $2 \le d \le 3$ and is closed for $d \ge 4$. However, as we observed in the case of $L=2$, irregular random LDLC may be useful to increase the upper limit for achievable $d$, e.g., for degree distribution $\lambda(t) = \rho(t) = 0.418683 \cdot t^2 + 0.162635 \cdot t^3 + 0.418683 \cdot t^5$ which has average degree $4$, the DPE tunnel is open whereas as per \refpro{pro_ec_d_L4}, it is closed for $4$-regular LDLC. With this, we propose a method for constructing LDLC for $L=4$ in the following theorem.
\begin{theorem} \label{thm_fulldiv_bp_L_4} 
Consider a block-fading channel with diversity $L=4$. Let $H=[h_{ij}]$ be the $n \times n$ integer-check matrix of a real Latin square LDLC $\Lambda$ of even rank $n$ and degree $3$ such that $n$ is divisible by four. 
Let us decompose $H$ into four $n/4 \times n$ submatrices as follows
\begin{align}
H = \left[
\begin{array}{c}
A \\ \hline
B \\ \hline
C \\ \hline
D
\end{array}
\right] \label{eq:fd_bp_H}.
\end{align}

\noindent
Assume that the binary image of $H$ has the following structure:
\[
H_b=\left[
\begin{array}{c c c c}
B_1 & \Pi_1 & 0 &  0 \\ \hline
 0  & B_2 & \Pi_2 &  0     \\ \hline
 0  & 0 & B_3 & \Pi_3 \\ \hline
 \Pi_4 & 0 & 0 & B_4 
\end{array}
\right].
\]
%
where $\Pi_k, k \in \{1,2,3,4\}$ are permutation matrices and $B_k, k \in \{1,2,3,4\}$ 
are random regular matrices with weight $2$ of dimension $n/4 \times n/4$. 
\\
Let $\theta_k, k \in \{1, 2, 3, 4\}$ be algebraic numbers of degree $\ge 2$
such that 
for 
$ p \ne q, \; \theta_p/\theta_q \notin \Q$. Then, the diversity order of the lattice defined by the
integer-check matrix
\begin{equation} \label{eq_H_FD_L_4}
\left[
\begin{array}{c}
\theta_1 A \\ \hline
\theta_2 B \\ \hline 
\theta_3 C \\ \hline
\theta_4 D
\end{array}
\right]
\end{equation}
is equal to the number of degrees of freedom in the channel, i.e., $L=4$.
%
\end{theorem}

The proof for \refthm{thm_fulldiv_bp_L_4} is similar to \refthm{thm_fulldiv_bp} and is omitted. 
We remark that the diversity order of the LDLC constructed as per \refthm{thm_fulldiv_bp} and \refthm{thm_fulldiv_bp_L_4} 
is same as the number of degrees of freedom in the channel under ML decoding. 
Further, it can be noted that the assumption $h_{i,j} \in \Q$ used in \refthm{thm_fulldiv_ml} and \refthm{thm_fulldiv_bp} is not necessary and the construction methods proposed in these theorems are valid even if $h_{i,j}$ is an irrational number.
%
%
%
%
%
%
%
%
%
%
%
%
%
%
%
%
%
%
%
%
\section{Simulation Results}
In this section, we report simulation results for LDLC constructed for diversity order $L=2$ and $L=4$ (for iterative decoding only). We use methods proposed in \refthm{thm_fulldiv_ml}, \refthm{thm_fulldiv_bp} and \refthm{thm_fulldiv_bp_L_4} to construct LDLC. In all simulations, we count at least $400$ erroneous points for each SNR value.

As mentioned in \refsec{sec_channel_model}, for a given SNR value $\gamma$, the POL gives an upper bound on $\sigma^2$ and consequently a lower bound on the value of $\prod_{l=1}^{L}\alpha_l$. 
If the value of  $\prod_{l=1}^{L}\alpha_l$ is below this limit (i.e., an inadmissible fading channel state) then even an optimal ML decoder would almost surely make a decoding error. 
Hence, it is possible for ML decoder to output an error without decoding when channel is in 
inadmissible fading state. This fact is utilized here to speed up the ML decoder of full-diversity LDLC. We remark that the POL described in \refsec{sec:POL} does not take finite dimensions of lattices into consideration and hence distance between $P_{\textrm{out}}(\gamma)$ and PER curves reported here is an lower bound. 

%

For ML decoding we utilize the integer-check matrix with dimension $n=64$ constructed according to the second matrix of \refeq{eq:h_fl_2} in \refthm{thm_fulldiv_ml}, where we select $\theta_1 = 1$ and $\theta_2 = \sqrt{2}$. We do not use any shaping region for the selected LDLC and decode using the ML decoder proposed in \cite{Viterbo1999}. 
We use a PC with Intel Xeon E5-2687W CPU clocked at 3.10 GHz.
Along with point error rate (PER) results for these LDLC, we also report results for total runtime required to complete simulations. 
%
%
%
\begin{figure*}
\vspace{-10pt}
\centering
\subfloat[Point Error Rate Performance without POL, with POL defined by $\prod_{l = 1}^{L} \alpha_l^2 =\frac{(2\pi e)^L}{\gamma^L}$,
and with POL+1.3 margin defined by $\prod_{l = 1}^{L} \alpha_l^2 ~=~1.3 \times \frac{(2\pi e)^L}{\gamma^L}$.]{
  \includegraphics[width=0.475\textwidth]{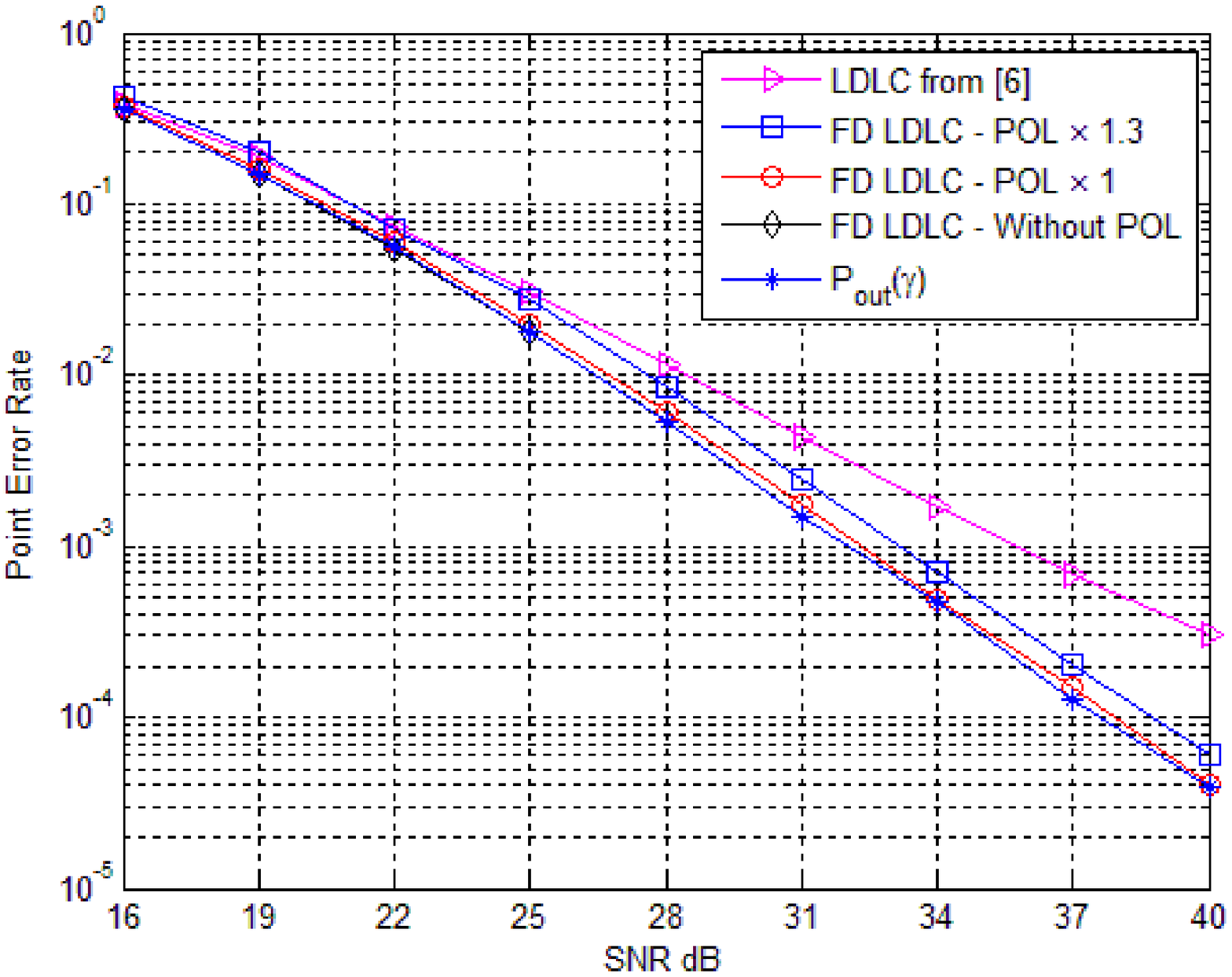}
}
\hfil
\subfloat[Runtime comparison at different values of signal-to-noise ratio. A huge POL gain in runtime is observed at $n=64$.]{
  \includegraphics[width=0.475\textwidth]{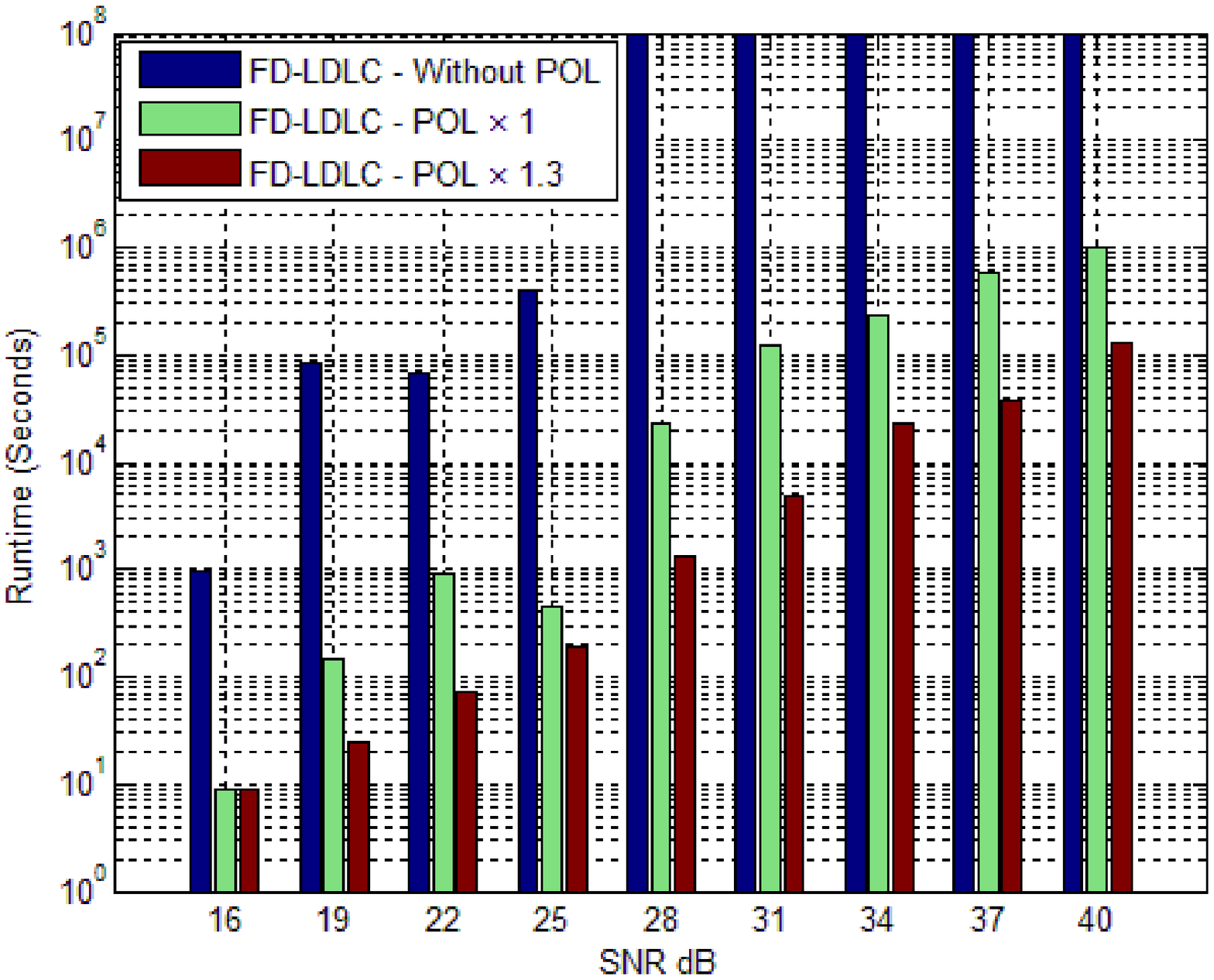}
}
\caption{ML decoding of double-diversity LDLC with dimension $n=64$.}
\label{fig:ML_64}
\vspace{-10pt}
\end{figure*}

Simulation results for the aforementioned LDLC are shown in \reffig{fig:ML_64}. We compare the point error rate and simulation runtime for programs that do not utilize POL with those where POL is utilized to declare error without decoding. Sphere decoder as proposed in \cite{Viterbo1999} is used for ML decoding. Four PER results are reported here : 
1) Random LDLC constructed as per Appendix VII in from \cite{Sommer2008}; 
2) FD-LDLC, POL is not used 
3) FD-LDLC, POL is used to detect inadmissible fading channel states; and 
4) FD-LDLC, POL is multiplied by a constant and the new value is used to detect inadmissible fading channel states.
$P_{\textrm{out}}(\gamma)$ is the POL for $L=2$ which also considers volume of lattices used here.
The curve for PER is parallel to that of POL and just $0.1$dB away from it.
When random LDLC is used over BF channel, the PER for such LDLC is not proportional to $\frac{1}{\gamma^2}$ at high SNR values and this fact can be observed from \reffig{fig:ML_64}(a). 
\begin{figure}[!t]
\centering
  \includegraphics[width=0.75\textwidth]{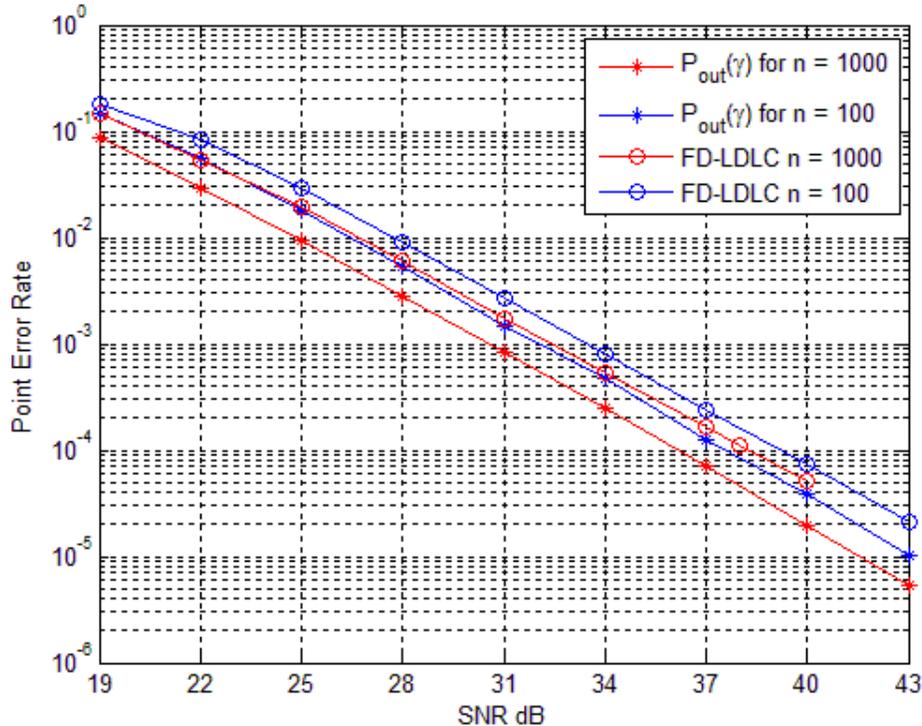}
\caption{Iterative decoding of double-diversity LDLC in dimension $n=100$ and $n=1000$. Point error rate is compared with respective POL.}
\label{fig:iter_L2}
\end{figure}

We compare accumulated runtime required by the decoder (until at least 400 erroneous points are obtained) in \reffig{fig:ML_64}(b). 
It can be observed that when either the POL (case 3) or a scaled POL (case 4) is used for simulations, we get significant improvements in terms of runtime. 
The impact on runtime of ML decoding which utilizes POL is clearly visible for the selected LDLC. For this LDLC, it is not possible to simulate the PER performance for SNR values higher than $25$dB without using POL in a feasible amount of time. 
The blue lines for SNR range $28$dB to $40$dB in \reffig{fig:ML_64}(b) are saturated to the maximum possible runtime and are shown for reference only. 
Also for SNR values less than $25$dB, 
the runtime with POL is only $10\%$ of the runtime required without using POL. The difference in runtime between case 3 (POL only) and case 4 (scaled POL) for this particular LDLC is also very high.

\begin{figure}[!t]
\centering
  \includegraphics[width=0.75\textwidth]{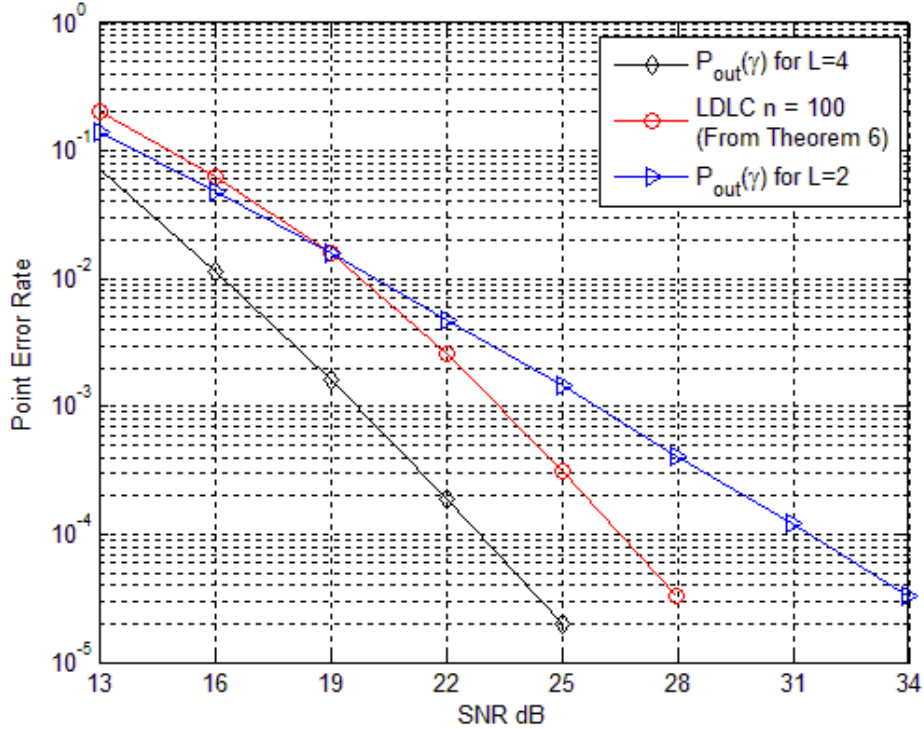}
\caption{Iterative decoding of LDLC constructed for diversity order $L=4$ in dimension $n=100$. Point error rate is compared with POL for $L=4$ and $L=2$.}
\label{fig:iter_L4}
\end{figure}

For the simulations using iterative decoding for $L=2$, we utilize two integer-check matrix constructed according to the first matrix of \refeq{eqn_h_bp} in \refthm{thm_fulldiv_bp}: 1)   $n = 100, d = 4$, where we select and $\theta_1 = 1$ and $\theta_2 = \frac{1}{\sqrt{2}}$, 2) $n = 1000, d = 5$ and $\theta_1 = 1$ and $\theta_2 = \sqrt{2}$. Generating sequence for both LDLC is $\left\{1, \frac{1}{\sqrt{5}}, \frac{1}{\sqrt{5}}, \frac{1}{\sqrt{5}}, \frac{1}{\sqrt{5}} \right\}$. Again, shaping region is not used during simulations.
We use the iterative decoder for LDLC proposed in \cite{Sommer2008} with several modifications which include a method to reduce complexity of the iterative decoder as suggested in Appendix VIII of \cite{Sommer2008}. 
We selected pdf length of $2^{16} = 65,536$ and FFT size of $1,024$ for the iterative decoder. 
Simulation results for these LDLC are given in \reffig{fig:iter_L2}. 
Two different curves for POL $\left(P_{\textrm{out}}(\gamma)\right)$ are reported in \reffig{fig:iter_L2} as the volume of two LDLC is different. The PER curves for both LDLC are parallel to the respective POL curves and hence these LDLC exhibit diversity order of $2$. The PER for $n = 100$ and $n=1000$ is around $1.5$dB and $2.1$dB away from the corresponding $P_{\textrm{out}}(\gamma)$.

Simulation results for the diversity order $L=4$ are given in \reffig{fig:iter_L4}. The LDLC used in this simulation is constructed using \refthm{thm_fulldiv_bp_L_4}. We selected a LDLC with $n=100, d = 3$, generating sequence $\left\{1, \frac{1}{\sqrt{3}}, \frac{1}{\sqrt{5}} \right\}$ and $\theta_1 = 1, \theta_2 = \frac{1}{\sqrt{2}}, \theta_3 = \frac{1}{\sqrt{3}}, \theta_4 = \frac{1}{\sqrt{7}}$. This LDLC is simulated using iterative decoder mentioned above which is adapted for $L = 4$ but with pdf length of $2^{18} = 262144$ and FFT size of $2048$. In this curve, PER is compared with the POL for channel diversity $L=2$ and $L=4$. The diversity order of $L=4$ for the aforementioned LDLC can be observed from \reffig{fig:iter_L4}. The PER for this LDLC is ca. $3$dB away from POL for $L=4$ however, it gains ca. $6$dB at SNR of $28$dB compared to POL for $L=2$.
\section{Conclusion and Future Work}
In this paper, first we defined a Poltyrev outage limit for lattices in presence of block fading. We proved that its diversity order is equal to the number of degrees of freedom in the BF channel. 
Further, we proposed construction methods for double-diversity lattices based on the integer-check matrix, the inverse of the lattice generator matrix. The first construction method is valid under ML decoding for both sparse and non-sparse integer-check matrices. 
The second construction method for double-diversity LDLC is based on sparse integer-check matrices and is valid for both, ML as well as iterative decoding. 
Further, one more method to construct LDLC with diversity order $4$ suitable for iterative decoding is proposed.
%
%
%

We also provided simulation results for ML decoding of double-diversity lattices and iterative decoding of lattices with diversity order $2$ and $4$. Our simulations results validate theoretical claims about the diversity order of the constructed lattices.
An important application of POL is lattice decoding of low-density lattice codes. 
The outage boundary of the POL is used during our simulations to declare an outage error without decoding, which makes the ML decoding of lattices on the BF channel tractable. 

The POL derived in this paper assumes lattices with infinite dimension. One future research direction could be to derive a POL for lattices with finite dimension. The work in this paper focuses on diversity of lattice only. Another possible future work could be to take into account the coding gain of low-density lattices. For example, for a given $L$, find the matrix weight $d$ that maximizes coding gain under ML and iterative decoding.

\appendices
\section{Proof of \refthm{thm_fulldiv_bp} (Full-diversity LDLC for iterative decoding)}\label{app_ldlc_iter}
 Let us assume that the lattice $\Lambda$ is constructed according to the matrix given on the left side of \refeq{eqn_h_bp} and $d=3$ without loss of generality. Further, we assume that the  lattice of \refeq{eq:fd_bp_H} is a Latin square LDLC. 
 Since the error probability of the iterative decoder for LDLC is independent of the transmitted codeword \cite{Sommer2008},  we can further assume that some random point $\bs{x} = (x_1, \ldots, x_n)$ is transmitted. 
 Let us select a check node $i_1 \in \{n/4+1,\ldots,n/2\}$ which is connected to the variable nodes $j_1, j_2, j_3$. 
 The integer-check equation for the check node $i$ for the point $\bs{x}$ can be written as follows
 \begin{equation} \label{eq_check_1}
  \theta_1 \left( h_{i_1 j_1} x_{j_1} + h_{i_1 j_2} x_{j_2} + h_{i_1 j_3} x_{j_3} \right) = z_{i_1} \;,
 \end{equation}
 where $z_{i_1} \in \mathbb{Z}$. For ease of exposition, the subscript 
 $j_1,j_2,j_3$ is replaced with $1,2,3$, respectively and $i_1$ is replaced by $1$. If \refeq{eq_check_1} is divided with $\theta_1 h_{13}$ then it can be written as
 \begin{equation}\label{eq_check_eq}
  x_{3} = \tilde{h}_{11} x_{1} + \tilde{h}_{12} x_{2} +  \tilde{h}_{13} \bar{\theta}_1 z_1\;,
 \end{equation}
  where $\tilde{h}_{11} = \frac{-h_{11}}{h_{13}}, \tilde{h}_{12} = \frac{-h_{12}}{h_{13}}, \tilde{h}_{13} = \frac{1}{h_{13}}$, and $\bar{\theta}_1 = \frac{1}{\theta_1} $. 

  \noindent
  Similarly, select check nodes $i_2 \in \{1,\ldots,n/4\}$ and $i_3 \in \{n/2+1,\ldots,(n/2)+(n/4)\}$ such that $x_3$ is connected to the check nodes $i_2$ and $i_3$. Then, similar to \refeq{eq_check_eq} we get following equations,
  \begin{align*}
    &x_{3} = \tilde{h}_{24} x_{4} + \tilde{h}_{25} x_{5} +  \tilde{h}_{23} \bar{\theta}_1 z_2, \\
    &x_{3} = \tilde{h}_{36} x_{6} + \tilde{h}_{37} x_{7} +  \tilde{h}_{33} \bar{\theta}_2 z_3,
  \end{align*}
  where $\tilde{h}_{24} = \frac{-h_{24}}{h_{23}}, \tilde{h}_{25} = \frac{-h_{25}}{h_{23}}, \tilde{h}_{23} = \frac{1}{h_{23}}$, $\tilde{h}_{36} = \frac{-h_{36}}{h_{33}}, \tilde{h}_{37} = \frac{-h_{37}}{h_{33}}, \tilde{h}_{33} = \frac{1}{h_{33}}$, and $\bar{\theta}_2 = \frac{1}{\theta_2}$. Here again, we dropped $i$ and $j$ for ease of exposition.

  \noindent
  Point $\bs{x}$ is transmitted over a BF channel with coefficients $\alpha_1$ and $\alpha_2$. 
  The construction proposed in \refeq{eqn_h_bp} enforces that, for a given check node, exactly one component out of $x_1, x_2, x_3$ would be affected by one of the channel coefficients, whereas the remaining two components would be affected by the other channel coefficient. Let us assume that the components $x_1, x_2$ are affected by $\alpha_1$ and $x_3$ is affected by $\alpha_2$. Similarly, $x_4, x_6$ are affected by $\alpha_1$ and fading coefficient for $x_5, x_7$ is $\alpha_2$. Then, the components for the received point $y$ can be given by 
  \begin{align}
   &y_1 = \alpha_1 x_1 + \eta_1, \; y_2 = \alpha_1 x_2 + \eta_2, \; y_3 = \alpha_2 x_3 + \eta_3, \label{eq_received_vec_1}\\
   &y_4 = \alpha_1 x_4 + \eta_1, \; y_5 = \alpha_2 x_5 + \eta_5, \; y_6 = \alpha_1 x_6 + \eta_6, \; y_7 = \alpha_2 x_7 + \eta_7. \label{eq_received_vec_2}
  \end{align}
  where $\eta_{k} \sim \N(t; 0, \sigma^2)$, $k \in \{1,\ldots,7\}$ is the noise. 

  \noindent
  The LDLC iterative decoder initializes the message of variable nodes with a probability density function (pdf) determined from the components of the received point. 
  The pdf for components $x_1, \ldots, x_7$ can be found as follows from \refeq{eq_received_vec_1} and \refeq{eq_received_vec_2},
  \begin{align*}
     f_k(t) \sim \N \left(t; \frac{y_k}{\alpha_1}, \frac{\sigma^2}{\alpha_1^2} \right), \; k=1,2,4,6; \quad
%
%
   f_k(t) \sim \N \left(t; \frac{y_k}{\alpha_2}, \frac{\sigma^2}{\alpha_2^2} \right), \; k = 3,5,7.
  \end{align*}
  During the first half of the first iteration, the check node $1$ sends pdf $p_{12}(t)$ to variable node $3$. 
  The message $p_{12}(t)$ is obtained by first calculating the convolution \footnote{The convolution of $n$ Gaussians with mean values $m_1, \ldots, m_n$ and variances $\sigma_1^2,\ldots,\sigma_n^2$, respectively, is a Gaussian with mean $m_1+\ldots+m_n$ and variance $\sigma_1^2+\ldots+\sigma_n^2$ \cite{Sommer2008}\cite{Pa_84}.} of Gaussians $f_1\left( {\tilde{h}_{11}}^{-1}{t}\right)$ and $f_2\left( {\tilde{h}_{12}}^{-1} {t}\right)$ and then by extending the result to a periodic function with period $\bar{\theta}_1 \tilde{h}_3$. Hence,
  \begin{equation} 
   p_{12}(t) = \sum_{z_1 = -\infty}^{\infty} \N \left(t; m_{12}(z_1),\sigma_{12} \right),
  \end{equation}
  where 
  \begin{align*}
    &m_{12}(z_1) = \frac{y_1 \tilde{h}_{11} + y_2 \tilde{h}_{12} }{\alpha_1 } - { \bar{\theta}_1 \tilde{h}_{13} z_1} + \eta_{12} = \tilde{h}_{11} x_1 + \tilde{h}_{12} x_2 - { \bar{\theta}_1 \tilde{h}_{13} z_1} + \eta_{12},\\
    &\eta_{12} = \frac{\eta_1 \tilde{h}_{1} + \eta_2 \tilde{h}_{2}}{\alpha_1}, \; \sigma^2_{12} = \frac{ \left( {\tilde{h}}^2_{11} + {\tilde{h}}^2_{12} \right) \sigma^2} {\alpha_1^2}.
  \end{align*}
  Similarly, messages from check nodes $2$ and $3$ to variable node $1$ are,
  \begin{equation*} 
   p_{45}(t) = \sum_{z_2 = -\infty}^{\infty} \N \left(t; m_{45}(z_2),\sigma_{45} \right), \; p_{67}(t) = \sum_{z_3 = -\infty}^{\infty} \N \left(t; m_{67}(z_3),\sigma_{67} \right),
  \end{equation*}
  where
  \begin{align*}
    &m_{45}(z_2) = \tilde{h}_{24} x_4 + \tilde{h}_{25} x_5 + {\tilde{h}_{23} \bar{\theta}_1 z_2} + \eta_{45},\\ %
    &\eta_{45} = \frac{\eta_4 \tilde{h}_{24} \alpha_2 + \eta_5 \tilde{h}_{25} \alpha_1}{\alpha_1 \alpha_2}, %
    \sigma^2_{45} = \frac{ \left( \alpha_2^2 \tilde{h}^2_{24} + \alpha_1^2 {\tilde{h}}^2_{25} \right) \sigma^2} {\alpha_1^2 \alpha_2^2},
    \end{align*}
    and
  \begin{align*}  
    &m_{67}(z_3) = \tilde{h}_{36} x_6 + \tilde{h}_{37} x_7 + {\tilde{h}_{33} \bar{\theta}_2 z_3} + \eta_{67},\\%
    &\eta_{67} = \frac{\eta_6 \tilde{h}_{36} \alpha_2 + \eta_7 \tilde{h}_{37} \alpha_1}{\alpha_1 \alpha_2}, %
    \sigma^2_{67} = \frac{ \left( \alpha_2^2 \tilde{h}^2_{36} + \alpha_1^2 {\tilde{h}}^2_{37} \right) \sigma^2} {\alpha_1^2 \alpha_2^2}.
  \end{align*}
  
  \noindent
  In second half of the first iteration, the variable node $3$ multiplies channel pdf $f_3(t)$ and check node messages $p_{12}(t), p_{45}(t)$ and $p_{67}(t)$ to generate the updated pdf $\tilde{f}_3(t)$ for the variable node $3$. An estimate of $x_3$, $\hat{x}_3$ is obtained by selecting the largest value in $\tilde{f}_3(t)$ \cite{Sommer2008}.
  The multiplication of Gaussian mixtures also generates a Gaussian mixture. Since check node messages are Gaussian mixtures, the multiplication of check node messages at variable node also generates a mixture of Gaussians. However, here we assume that only the largest peak is retained from this mixture whereas other peaks are attenuated to zero. Due to this operation, the resulting decoder is a sub-optimal iterative decoder. 
%
  Exact expressions for the mean and variance of multiplication of Gaussian densities can be found in Claim 1 and Claim 2 of \cite{Sommer2008}. 

  \noindent
  We now assume without loss of generality that the transmitted point $x$ is the all-zeros point and prove that after multiplication and removal of smaller peaks, the only remaining peak in $\tilde{f}_3(t)$ is associated with $z_1 = z_2 = z_3 = 0$. For this, we derive the amplitude of one of the peaks of $\tilde{f}_3(t)$ \cite{Sommer2008} by selecting peaks associated with $m_{12}(z_1), m_{45}(z_2)$ and $m_{67}(z_3)$.
  \begin{align}
   &\N(t; m_3, \sigma_3^2) \cdot \N(t; m_{12}(z_1), \sigma_{12}^2) \cdot \N(t; m_{45}(z_2), \sigma_{45}^2)  \cdot \N(t; m_{67}(z_3), \sigma_{67}^2) \nonumber \\
   &= \hat{A}(z_1,z_2,z_3) \cdot \N (t; \hat{m}_3(z_1,z_2,z_3), \hat{\sigma}_3^2), \label{eq_mul_Gaussians}
  \end{align}
  where $\hat{A}(z_1,z_2,z_3)$, $\hat{m}_3(z_1,z_2,z_3)$, and $\hat{\sigma}_3^2$ is amplitude, mean, and variance of the Gaussian obtained by multiplying $m_{12}(z_1), m_{45}(z_2)$ and $m_{67}(z_3)$, respectively.
  Exact expression for \\ $\hat{A}(z_1,z_2,z_3)$ is given in \refeq{eq_A_hat}. As can be observed from \refeq{eq_A_hat}, we have 
  \begin{align*}
   x_3 \; = \; \tilde{h}_{11} x_1 + \tilde{h}_{12} x_2 + \bar{\theta}_1 \tilde{h}_{13} z_1
   \; = \; \tilde{h}_{24} x_4 + \tilde{h}_{25} x_5 + \bar{\theta}_1 \tilde{h}_{23} z_2
   \; = \; \tilde{h}_{36} x_6  + \tilde{h}_{37} x_7 + \bar{\theta}_2 \tilde{h}_{33} z_3, 
  \end{align*}
  if and only if $z_1 = z_2 = z_3 = 0$, which in turn maximizes $\hat{A}(z_1,z_2,z_3)$. If we select ${\theta_1}$ and ${\theta_2}$ such that $\frac{\theta_1}{\theta_2} \in \Q$ then either more than one set of values of $z_1, z_2, z_3$ or a single set of values other than $z_1 = z_2 = z_3 = 0$ may maximize $\hat{A}(z_1, z_2, z_3)$. Hence, it is necessary to select $\frac{\theta_1}{\theta_2} \notin \Q$ so that only an appropriate set of values of $z_1, z_2, z_3$ maximizes $\hat{A}(z_1, z_2, z_3)$.
  In some cases the noise terms, e.g, $\sigma_{3} - \sigma_{12}$, etc., in \refeq{eq_A_hat} may be high enough to shift optimal solution away from $z_1 = z_2 = z_3 = 0$. However, we note that the effect of noise in \refeq{eq_A_hat} on diversity order of the lattice under consideration can be ignored; reasons for which are explained at the end of this proof. Now, we prove that the lattice considered here indeed have diversity order $2$.

  \begin{figure*}[!t]
  \normalsize
  \begin{align}
   &\hat{A}(z_1,z_2,z_3) \propto \exp\left( -\frac{ \hat{\sigma_3} }{ 2 } \left( \frac{\left( m_3 - m_{12}(z_1) \right)^2 }{ \sigma_3 \sigma_{12} }%
   +  \frac{\left(m_3 - m_{45}(z_2)\right)^2}{\sigma_3 \sigma_{45}} + \left( \frac{m_3 - m_{67}(z_3) }{ \sigma_3 \sigma_{67} } \right)^2 
   \right. \right. \nonumber \\
   &\quad\quad\quad\quad\quad\quad\quad\quad+\frac{ \left( m_{12}(z_1) - m_{45}(z_2) \right)^2}{ \sigma_{12} \sigma_{45} }  + \frac{ \left( m_{12}(z_1) - m_{67}(z_3) \right)^2}{\sigma_{12} \sigma_{67} }  
   \nonumber \\
   &\quad\quad\quad\quad\quad\quad\quad\quad
   +\left.\left.\frac{ \left( m_{45}(z_2) - m_{67}(z_3) \right)^2}{\sigma_{45} \sigma_{67}}  \right) \right). \nonumber \\
%
%
%
%
%
   &= \exp\left( -\frac{\hat{\sigma_3}}{2} \left( \frac{ \left( x_3 - (\tilde{h}_{11} x_1 + \tilde{h}_{12} x_2 + \bar{\theta}_1 \tilde{h}_{13} z_1) + \eta_{3}-\eta_{12}  \right)^2}{\sigma_3 \sigma_{12}}%
   \right. \right. 
   \nonumber \\
   &\quad\quad\quad
   +\frac{ \left( x_3 - (\tilde{h}_{24} x_4 + \tilde{h}_{25} x_5 + \bar{\theta}_1 \tilde{h}_{23} z_2) + \eta_{3}-\eta_{45} \right)^2 }{\sigma_3 \sigma_{45}} 
   \nonumber \\
   &\quad\quad\quad
   +\frac{ \left( x_3 - (\tilde{h}_{36} x_6 + \tilde{h}_{37} x_7 + \bar{\theta}_2 \tilde{h}_{33} z_3) + \eta_{3}-\eta_{67} \right)^2 }{\sigma_3 \sigma_{67}}  
   \nonumber \\
  &\quad\quad\quad
  +\frac{ \left( (\tilde{h}_{11} x_1 + \tilde{h}_{12} x_2 + \bar{\theta}_1 \tilde{h}_{13} z_1) - (\tilde{h}_{24} x_4 + \tilde{h}_{25} x_5 + \bar{\theta}_1 \tilde{h}_{23} z_2) + \eta_{12}-\eta_{45} \right)^2 }{\sigma_{12} \sigma_{45}} \nonumber \\
  &\quad\quad\quad
  +\frac{ \left( (\tilde{h}_{11} x_1 + \tilde{h}_{12} x_2 + \bar{\theta}_1 \tilde{h}_{13} z_1) - (\tilde{h}_{36} x_6 + \tilde{h}_{37} x_7 + \bar{\theta}_2 \tilde{h}_{33} z_3) + \eta_{12}-\eta_{67} \right)^2 }{\sigma_{12} \sigma_{67}} 
  \nonumber \\
  &\quad\quad\quad
  \left. \left. + \frac{ \left( (\tilde{h}_{24} x_4 + \tilde{h}_{25} x_5 + \bar{\theta}_1 \tilde{h}_{23} z_2) - (\tilde{h}_{36} x_6 + \tilde{h}_{37} x_7 + \bar{\theta}_2 \tilde{h}_{33} z_3) + \eta_{45}-\eta_{67} \right)^2}{\sigma_{45} \sigma_{67}}  \right) \right) \label{eq_A_hat}%
  \end{align}
%
%
  \hrulefill
  \vspace*{4pt}
  \end{figure*}  

  \noindent
  The value of $\hat{m}_3(z_1, z_2, z_3)$ in \refeq{eq_mul_Gaussians} for $z_1 = z_2 = z_3 = 0$ can be calculated as follows,
  \begin{align}
    &\hat{m}_3(z_1 = 0, z_2 = 0, z_3 = 0) \propto \frac{m_3}{\sigma_3^2} + \frac{m_{12}(z_1=0)}{\sigma_{12}^2} + \frac{m_{45}(z_2=0)}{\sigma_{45}^2} + \frac{m_{78}(z_3=0)}{\sigma_{78}^2} \nonumber
    \\
    &= \frac{x_3 \alpha_2^2}{\sigma^2} %
    + \frac{\left(\tilde{h}_1 x_1 + \tilde{h}_2 x_2 + \eta_{12} \right) \alpha_1^2}{\left( {\tilde{h}}^2_{11} + {\tilde{h}}^2_{12} \right) \sigma^2} %
    + \frac{\left(\tilde{h}_4 x_4 + \tilde{h}_5 x_5 + \eta_{45} \right) \alpha_1^2 \alpha_2^2 } { \left( \alpha_2^2 \tilde{h}^2_{24} + \alpha_1^2 \tilde{h}^2_{25} \right)  \sigma^2 } \nonumber\\%
    &\quad+ \frac{ \left(\tilde{h}_6 x_6 + \tilde{h}_7 x_7 + \eta_{67} \right) \alpha_1^2 \alpha_2^2 }{ \left( \alpha_2^2 \tilde{h}^2_{36} + \alpha_1^2 \tilde{h}^2_{37} \right) \sigma^2 } \nonumber
  \end{align}
  which gives,
  \begin{align}
    \hat{m}_3(z_1 = 0, z_2 = 0, z_3 = 0) \propto \left( \omega_1 \alpha_1^2 + \omega_2 \alpha_2^2 + \omega_3 \alpha_1^2 \alpha_2^2 \right) x_3 + \eta_3^{\prime} \label{eq_m_3}
  \end{align}
  where,
  \begin{align*}
   &\omega_1 = \frac{1}{\left( {\tilde{h}}^2_{11} + {\tilde{h}}^2_{12} \right) \sigma^2}, \quad%
   \omega_2 = \frac{1}{\sigma^2}, \quad%
   \omega_3 = \frac{1}{ \left( \alpha_2^2 \tilde{h}^2_{24} + \alpha_1^2 \tilde{h}^2_{25} \right)  \sigma^2 } + \frac{1}{ \left( \alpha_2^2 \tilde{h}^2_{36} + \alpha_1^2 {\tilde{h}}^2_{37} \right) \sigma^2 }, \\
   &\eta_3^{\prime} = \frac{\alpha_1^2 \eta_{12}}{\omega_1} + \frac{\alpha_2^2 \eta_{3}}{\omega_2} + \frac{\alpha_1^2 \alpha_2^2 \left(\eta_{45} + \eta_{67}\right) }{\omega_3}
  \end{align*}
  We would like to decide $x_3$ from $\hat{m}_3$ given in \refeq{eq_m_3} but this equation informs us that this decision has‏ error rate behaving like $1/\gamma^2$ (See \cite[Sec. 13.4]{Proakis2008}) because of the generalized $\chi^2$ distribution of order $4$ in ‏ $(\omega_1 \alpha_1^2 + \omega_2 \alpha_2^2 + \omega_3 \alpha_1^2 \alpha_2^2)$. The fourth order $\chi^2$ distribution guarantees the double diversity for the sub-optimal iterative decoder. Since the sub-optimal iterative decoder achieves full diversity, the original iterative decoder is also guaranteed to achieve it.

  \noindent
  We stated earlier that the effect of noise in \refeq{eq_A_hat} on diversity order of the lattice can be ignored. There are two reasons for this; first, as mentioned in \refsec{sec_channel_model}, the diversity order of a lattice is defined for $\gamma \to \infty$ which means $\sigma^2 \to 0$. Hence, the effect of noise in \refeq{eq_A_hat} on diversity order of the lattice is negligible. To explain second reason we examine the error probability at the output of decoder, 
  \begin{align}
    P_e \le P(z = 0) \frac{1}{\gamma^2} + P(z \ne 0) = P(z = 0) \frac{1}{\gamma^2} + e^{-\beta \gamma}, \label{eq_P_e}
  \end{align}
  where $\beta$ is an appropriate constant. The error probability in \refeq{eq_P_e} is dominated by the first term which is related to the diversity order of the lattice. The second term has negligible effect on $P_e$ at high SNR and therefore, the effect of noise in \refeq{eq_A_hat} on diversity order can be ignored.

  \noindent
  The above analysis assumes $z = 0$. However, it remains valid also for any $z \in \mathbb{Z}^n\setminus\{0\}^n$ as for any other value of $z$, only $\omega_1, \omega_2, \omega_3$ and/or $\eta^{\prime}_{3}$ in \refeq{eq_m_3} would be affected and the $\chi^2$ distribution of order 4 appearing in \refeq{eq_m_3} would remain intact.

\ifCLASSOPTIONcaptionsoff
  \newpage
\fi

\begin{thebibliography}{100}

\bibitem{Proakis2008}
J.G.~Proakis and M.~Salehi, {\em Digital Communications}, McGraw-Hill, 5th edition, 2008.

\bibitem{Conway1998}
J.H.~Conway and N.J.A.~Sloane, {\em Sphere Packings, Lattices and Groups}, 
3rd edition, Springer-Verlag, New York, 1998.

\bibitem{Poltyrev1994}
G. Poltyrev, 
``On coding without restrictions for the AWGN channel,'' 
{\em IEEE Transactions on Information Theory}, 
vol.~40, pp.~409-417, March 1994.

\bibitem{Cover2006} T.M. Cover and J.A. Thomas, {\em Elements of Information Theory},
Wiley, 2nd edition, 2006.

\bibitem{Ingber2013} 
A.~Ingber, R.~Zamir and M.~Feder, 
``Finite Dimensional Infinite Constellations,''
{\em IEEE Transactions on Information Theory}, 
vol. 59, no. 3, pp. 1630-1656, March 2013.

\bibitem{Leech1971}
J. Leech and N.J.A. Sloane, ``Sphere packing and error-correcting codes,'' Canadian Journal of Mathematics, no. 23, pp. 718-745, 1971.

\bibitem{Boutros1996}
J.~J. Boutros, E. Viterbo, C. Rastello and J.-C. Belfiore, ``Good lattice constellations for both Rayleigh fading and Gaussian channels,'' IEEE Transactions on Information Theory, vol. 42, no. 2, pp. 502-518, March 1996.

\bibitem{Sommer2008}
N. Sommer, M. Feder, and O. Shalvi, 
``Low-Density Lattice Codes,'' 
{\em IEEE Transactions on Information Theory}, 
vol. 54, no. 4, pp 1561-1585, April 2008.

\bibitem{Boutros2014}
J.J.~Boutros, N.~di Pietro, and N.~Basha,
``Generalized low-density (GLD) lattices,''
{\em Proc. of the 2014 IEEE Information Theory Workshop}, 
  Hobart, pp.~15-19, Nov. 2014.

\bibitem{Yona2014}
Y.~Yona and M.~Feder, 
``Fundamental Limits of Infinite Constellations in MIMO Fading Channels,'' 
{\em IEEE Transactions on Information Theory}, 
vol. 60, no.2, pp. 1039-1060, Feb. 2014.

\bibitem{diPietro2012}
N. di Pietro, J.J. Boutros, G. Z\'emor, and L. Brunel,
``Integer low-density lattices based on Construction A,''
{\em Proc. of the IEEE Information Theory Workshop}, 
pp.~422-426, Lausanne, Sept. 2012.

\bibitem{diPietro2013}
N. di Pietro, G. Z\'emor, and J.J. Boutros
``New results on Construction A lattices based on very sparse parity-check matrices,''
{\em Proc. of the IEEE Intern. Symp. on Inf. Theory (ISIT)}, 
pp.~1675-1679, July 2013.


\bibitem{Viterbo1999} 
E. Viterbo and J.~J. Boutros,``A universal lattice code decoder for fading channels,'' {\em IEEE Transactions on Information Theory}, vol.~45, no.~5, pp.~1639-1642, July 1999.

\bibitem{Bender1978} C.~M.~Bender and S.~A.~Orszag, ``Asymptotic expansion of integrals.'' in {\em Advanced Mathematical Methods for Scientists and Engineers}, McGraw-Hill Book Company, 1978, pp. 251-252.

\bibitem{Tse2005} 
D.N.C.~Tse and P.~Viswanath, {\em Fundamentals of Wireless Communication}. Cambridge University Press, 2005.

\bibitem{Richardson2008} 
T.J.~Richardson and R.L.~Urbanke, {\em Modern Coding Theory}. Cambridge University Press, 2008.

\bibitem{Boutros2009ita} J.J. Boutros, ``Diversity and coding gain evolution in graph codes,''
{\em Information Theory and Applications}, pp. 34-43, UCSD, San Diego, Feb. 2009.

\bibitem{Boutros2009}
J.J.~Boutros, A.~Guill\'en i F\`abregas, E.~Biglieri, and G.~Z\'emor,
``Low-density parity-check codes for nonergodic block-fading channels,'' 
{\em IEEE Transactions on Information Theory}, vol.~56, no.~9, pp.~4286-4300, Sept. 2010.

\bibitem{Pa_84}
A.~Papoulis, {\em Probability, Random Variables and Stochastic Processes}. New York: McGraw-Hill, 2nd edition, 1984.




















%
%
\end{thebibliography}
%
%

%




\end{document}